\documentclass{article}
\usepackage{amssymb}

\usepackage{etex}
\usepackage{amsmath}
\usepackage{tikz}
\usepackage{tkz-graph}
\usepackage{bussproofs}
\usepackage{graphicx}
\usepackage{color}

\usetikzlibrary{arrows,automata}
\usetikzlibrary{arrows,shapes,positioning}
\usetikzlibrary{tikzmark,decorations.pathreplacing,calc}
\newtheorem{theorem}{Theorem}

\newtheorem{claim}[theorem]{Claim}

\newtheorem{corollary}[theorem]{Corollary}

\newtheorem{definition}[theorem]{Definition}
\newtheorem{example}[theorem]{Example}

\newtheorem{lemma}[theorem]{Lemma}

\newtheorem{partial solution}[theorem]{Partial Solution}

\newenvironment{proof}[1][Proof]{\textbf{#1.} }{\ \rule{0.5em}{0.5em}}

\input{tcilatex}

\begin{document}

\begin{center}
{\Large Proof of NP = coNP = PSPACE: Current upgrade\medskip }

L. Gordeev, E. H. Haeusler

\textit{Universit\"{a}t T\"{u}bingen, PUC Rio de Janeiro}

\texttt{lew.gordeew@uni-tuebingen.de}\textit{,}

\texttt{hermann@inf.puc-rio.br\medskip }
\end{center}

\textbf{Abstract. }{\small In \cite{GH2} (see also \cite{GH1}, \cite{GH3})
we skechted proofs of the equalities \textbf{NP} = \textbf{coNP} = \textbf{%
PSPACE}. Recall that relations between these well-known complexity classes
of computational problems are usually referred to as fundamental open
problems in theoretical computer science (cf. e.g. \cite{arora}, \cite{Papa}%
). It is known, however, that \textbf{PSPACE} can be characterized by the
provability in minimal propositional logic (\cite{Statman}, \cite{Svejdar}),
and \textbf{NP} $\subseteq $ \textbf{PSPACE} , while \textbf{NP} = \textbf{%
coNP} is a consequence of \textbf{NP} = \textbf{PSPACE} (see e.g. \cite
{arora}, \cite{Papa}). Our proofs use these results -- in fact, by the last
one it suffices to prove only the inclusion \textbf{PSPACE} $\subseteq $ 
\textbf{NP}. We apply novel tree-to-dag proof
compressing techniques adapted to Prawitz's \cite{Prawitz} Natural Deduction
(ND) for the propositional minimal logic coupled with corresponding Hudelmaier's
cutfree sequent calculus \cite{Hud}. This enables to compress ``huge''
(size-exponential) tree-like ND proofs into equivalent ``small''
(size-polynomial) dag-like proofs with respect to a modified (called \emph{%
regular} below) notion of provability that has polynomial certificate. This
yields \textbf{PSPACE} $\subseteq $ \textbf{NP}, as desired; in the present
paper we propose a more transparent and careful presentation. We also
comment on a different approach presented in \cite{Jer} that claimed to
obtain exponential lower bounds on the size of dag-like ND proofs under a
more conventional notion of provability determined by the condition \emph{%
``all deductive paths are closed''}, instead of our weakening \emph{``all
regular deductive paths are closed''}. We show (\S 1.3 below) that the
former notion of provability is stronger than ours, and hence 
\cite{Jer} is irrelevant to our results. That is, exponential
lower bounds from \cite{Jer} are well compatible with our polynomial upper
bounds obtained for ND proofs under our ``regular'' provability. Moreover,
according to \cite{Jer}, these exponential lower bounds should hold also for
designated Frege systems involved, which therefore appear useless for the
sake of comparison \textbf{NP} vs \textbf{coNP} vs \textbf{PSPACE}. To put
it more precisely, there are polynomial dag-like proofs in our ND formalism
which have no analogy in that of Frege systems from \cite{Jer}. }\footnote{%
{\footnotesize In fact, \cite{Jer} claims exponential lower bounds for the entire 
implicational minimal logic. This conclusion is wrong, as 
demonstrated by our Basic example (\S 1.3) and Example 8 (\S 2.2).}}

\section{Introduction}

\subsection{Complexity classes}

Recall standard definitions and well-known results concerning complexity
classes \textbf{NP}, \textbf{coNP} and \textbf{PSPACE} (cf. e.g. \cite{arora}%
, \cite{Papa}, \cite{Statman}, \cite{Svejdar}).

A language $L\subseteq \left\{ 0,1\right\} ^{\ast }$\ is in \textbf{NP}
(resp. \textbf{coNP}) if there exists a polynomial $p$ and a polynomial-time
TM $M$ such that a given input $x\in \left\{ 0,1\right\} ^{\ast }$ is in $L$
iff $M\left( x,u\right) =1$ for some (resp. every) $u\in \left\{ 0,1\right\}
^{\ast }$ of the length $\left| u\right| \leq p\left( \left| x\right|
\right) $. Note that \textbf{coNP} is complementary to \textbf{NP} (and vice
versa). However it is unclear a priori whether\ symmetric difference $\left( 
\mathbf{NP\setminus coNP}\right) \mathbf{\cup }\left( \mathbf{coNP\setminus
NP}\right) $ is empty or not. In the former case we'll arrive at \textbf{NP}
= \textbf{coNP}, which looks more natural and/or plausible, as it reflects
an idea of logical equivalence between model theoretical (re: \textbf{NP})
and proof theoretical (re: \textbf{coNP}) interpretations of
non-deterministic polynomial-time computability.

Now $L\subseteq \left\{ 0,1\right\} ^{\ast }$ is in \textbf{PSPACE} if there
exists a polynomial $p$ and a TM $M$ such that for every input $x\in \left\{
0,1\right\} ^{\ast }$, the total number of non-blank locations that occur
during $M$'s execution on $x$ is at most $p\left( \left| x\right| \right) $,
and $x$ is in $L$ iff $M\left( x\right) \!=\!1$. Thus \textbf{PSPACE}
requires polynomial upper bounds only on the space -- but not time -- of
entire computation. It is well-known that \textbf{NP} and \textbf{coNP} are
both contained in \textbf{PSPACE}, but a priori unclear whether at least one
of these classes is a proper subclass of \textbf{PSPACE}. It is known,
however, that the assumption \textbf{NP} = \textbf{PSPACE} implies \textbf{NP%
} = \textbf{coNP}, while validity problem in propositional minimal logic,
VAL, is representable by a language $L\subseteq \left\{ 0,1\right\} ^{\ast }$
from \textbf{PSPACE}.

Keeping this in mind we ask how to solve our crucial problem

$\left( Q\right) $: Is \textbf{PSPACE} contained in \textbf{NP} or not?

To prove the affirmative assumption it suffices to find a proof system 
$T$ in minimal logic whose provability assertion ``derivation $\partial $
proves formula $\varphi $'' has polynomial certificate (i.e. is verifiable
by a TM in time polynomial in the weight of $\partial $) and such that every
formula $\varphi \in VAL$ is provable by a derivation $\partial $ whose
weight is polynomial in the size of $\varphi $. To this end we let $T$ be a
suitable dag-like modification NM{\small $_{\rightarrow }^{\text{\textsc{d}}}
$} of Prawitz's ND for purely implicational minimal logic (see below) and
prove the required assertions. Our proof system NM{\small $_{\rightarrow }^{%
\text{\textsc{d}}}$} is weaker (though sound and complete for minimal logic)
than more familiar tree-like ND presentations and conventional dag-like
versions thereof. The main difference is based on our modified notion of
dag-like provability that is determined by the requirement ``all \emph{%
regular} deductive paths are closed'', for a suitable condition of
regularity involved. Another crucial advantage of NM{\small $_{\rightarrow
}^{\text{\textsc{d}}}$}\ \ follows directly from its dag-like structure that
enables us to compress given ``huge'' tree-like derivations into the desired
``small'' dag-like ones.

By contrast, it is unclear how to prove $\left( Q\right) $'s negation,
should the affirmative assumption fail. To conclude that a given $L\subseteq \left\{
0,1\right\} ^{\ast }$ is not in \textbf{NP} it would not suffice to prove
superpolynomial lower bounds on proofs in a chosen proof system $T$ because
no explicit description of all polynomially equivalent proof systems
involved is currently known. For instance, an attempt to address so-called
Frege systems (cf. \cite{Jer} and attached references) fails, as
demonstrated in this paper by Basic example and Example 8 (see \S \S 1.3,
2.2 below) showing that our proof system NM{\small $_{\rightarrow }^{\text{%
\textsc{d}}}$} has no analogy in \cite{Jer}. Indeed, \cite{Jer} deals only
with proof systems claimed to obey exponential lower bounds for proofs of
minimal tautologies, whereas NM{\small $_{\rightarrow }^{\text{\textsc{d}}}$}
obeys polynomial upper bound in question. Thus exponential lower bounds
claimed in \cite{Jer}, though interesting, are not relevant to our results,
and hence Frege systems and equivalent natural deductions from \cite{Jer}
appear useless for $\left( Q\right) $.

\subsection{Logic and proof systems}

In early 20th century \emph{proofs} (a.k.a. \emph{deductions} or \emph{
derivations}) were understood in linear form: \emph{Formula $\varphi $ is
derivable from axioms $\mathcal{A}$ (abbr.: $\mathcal{A}\vdash \varphi $)
iff 
\begin{equation*}
\left( \exists \varphi _{0},\cdots ,\varphi _{n}\right) \left( \varphi
=\varphi _{n}\ and\ \left( \forall k\leq n\right) \varphi _{k}\in A\ and\
\left( \exists i,j<k\right) \varphi _{j}=\varphi _{i}\rightarrow \varphi
_{k}\right) \text{,}
\end{equation*}
i.e. $\varphi _{k\ }$follows from $\varphi _{i}$ and $\varphi _{j}$ by the
rule ``modus ponens'' (detachment)}. Other sound rules of inference were
treated the same way. . These definitions reflect \emph{linear} presentation
of mathematical proofs of new theorems via axioms, known theorems, suitable
lemmas, etc. Corresponding proof systems in mathematical logic are called
Hilbert-Bernays style calculi (due to \cite{HB1}, \cite{HB2}, see also \cite
{Neu}, \cite{Tar}); certain extended versions are also called Frege systems
(cf. e.g. \cite{Jer} and attached references). Proofs in the algebraic logic
(boolean and relation algebras) are understood analogously with regard to
standard equational calculus and the transitivity rule of ``$=$'', instead
of \emph{modus ponens }(see e.g. \cite{TG} for references)

Graph-theoretic interpretations of logic proofs lead to genuine \emph{%
structural} proof theory, where basic proof systems are \emph{natural
deduction} (ND) and \emph{sequent calculus} (SC) -- initiated by Gentzen 
\cite{Gen1}, \cite{Gen2} and further developed in \cite{Prawitz} resp. \cite
{Sch1}, \cite{Sch2}, \cite{TS}, \cite{Hud}, \cite{Fel} (et al). Here proofs
are usually presented in tree-like form, where branching points are
determined by the conclusions of (at least binary) rules involved. ND proofs
contain single formulas, unlike SC proofs containing \emph{sequents}
instead. Note that ND proofs have no axioms. Instead, all assumptions shall
be \emph{discharged} according to special conditions on the threads. Both ND
and SC allow \emph{normalizations} (mutually different) making proofs more
transparent and suitable for analysis.

Linear proofs admit tree-like interpretation, and v.v.. Different nodes in
tree-like proofs might correspond to identical formulas (``references'') $%
\varphi _{i}$, $\varphi _{j}$ occurring in linear proofs $\mathcal{A}\vdash
\varphi $ (sequent case is analogous). So passing to tree-like proofs might
essentially increase the size of linear inputs. The opposite direction is
called \emph{proof compression}. Actually we compress tree-like ND proofs
into \emph{dag-like} ND proofs (\emph{dag} = directed acyclic graph) by
merging different nodes labeled with identical formulas.

\emph{Classical propositional logic} provides natural interpretations of $%
\mathbf{NP}$ and $\mathbf{coNP}$. Namely, well-known propositional
satisfiability and validity problems SAT and VAL are, respectively, $\mathbf{%
NP}$- and $\mathbf{coNP}$-complete. That is, any given language encoding
satisfiable (resp. valid) formulas is universal for the whole class $\mathbf{%
NP}$ (resp. $\mathbf{coNP}$). However, classical proof systems correlated
with VAL are less helpful for the comparison $\mathbf{NP}$ vs $\mathbf{coNP}$%
, as the size of related proofs of classical tautologies $x$ tends to be
exponential in $\left| x\right| $. \emph{Minimal} \emph{propositional logic} 
\cite{Joh}, \cite{PraMa} provides more suitable refinements. So we replace
classical logic by the minimal one while keeping in mind polynomial-size
interpretation of the former. Our chosen proof system (call it $T$) is a
Prawitz-style ND calculus adapted to minimal logic. This choice is not new,
except for our notion of the dag-like provability. To put it more exactly,
conventional interpretation of $T$\ (or other ND) assumes that deductions $%
\partial \in T$ are rooted trees whose nodes are assigned with formulas
ordered according to the inference rules allowed in $T$ (: locally correct
deductions). Top formulas and the root one are called assumptions and
conclusion of $\partial $, respectively. Proofs are locally correct
deductions whose all assumptions are discharged (cf. \cite{Prawitz}). We
consider a more liberal interpretation allowing dag-like deductions $%
\partial \in T$ ordered as \emph{dags} (directed acyclic graphs) -- not
necessarily trees. Such dag-like proofs require a more sophisticated notion
of provability (cf. \S 1.2 below). Obviously dag-like deductions tend be
smaller (even exponentially) than corresponding tree-like ones. We
elaborated special \emph{horizontal compressing} techniques that allow to
compress ``huge'' tree-like proofs $\partial \in T$ into ``small'' dag-like
proofs $\partial _{0}\in T$ with the same conclusions such that proof
correctness of $\partial _{0}$ is verifiable in polynomial time by
deterministic TM.

\subsection{More on tree- vs dag-like provability}

Propositional minimal logic is determined by two axioms $\alpha \rightarrow
\left( \beta \rightarrow \alpha \right) $ and $\left( \alpha \rightarrow
\left( \beta \rightarrow \gamma \right) \right) \rightarrow \left( \left(
\alpha \rightarrow \beta \right) \rightarrow \left( \alpha \rightarrow
\gamma \right) \right) $\ and rule \emph{modus ponens} $\fbox{$\dfrac{\alpha
\quad \quad \alpha \rightarrow \beta }{\beta }$}$ in the Hilbert-style
formalism whose vocabulary includes propositional variables and one
propositional connective `$\rightarrow $' (\emph{implication}; $\alpha $, $%
\beta $, $\gamma $, etc. are arbitrary formulas). The intuitionistic logic
extends minimal logic by one constant $\bot $ (falsity) and new axiom $\bot
\rightarrow \alpha $. It is well-known that there are polynomial-size
validity-preserving embeddings of classical into intuitionistic and
intuitionistic into minimal logic, respectively. Instead of Hilbert-style
systems for minimal and intuitionistic logic, we consider Gentzen-style 
\emph{sequent calculus} (SC) and Prawitz-style \emph{natural deductions}
(ND). These proof systems admit well-known proof-optimization: \emph{cut
elimination} in SC and \emph{normalization} in ND.

Cut elimination provides sound and complete systems of inferences without 
\emph{cut rule} (that is equivalent to modus ponens). Inferences in
resulting \emph{cutfree} SC systems satisfy a sort of \emph{subformula
property} (: all premise formulas occur as (sub)formulas in the
conclusions), which enables better proof search strategies. We can also
assume that the heights of cutfree proofs are linear in the size of
conclusions (this is not obvious for intuitionistic and/or minimal logic,
see \cite{Hud}, \cite{GH1}). In ND, the normalization allows to use just 
\emph{normal} proofs that are known to satisfy \emph{weak subformula property%
} (: every formula occurring in a maximal path occurs as (sub)formula in the
conclusion), see \cite{Prawitz}. However, there are no polynomial upper
bounds on the heights of arbitrary normal deductions in ND.

These optimizations have been elaborated for standard tree-like versions of
SC and ND. Note that tree-like approach can't provide polynomial upper
bounds on the size of resulting proofs. To achieve this goal we formalize
another idea called \emph{horizontal compression}. That is, in a given
tree-like proof we merge all nodes labeled with identical objects (sequents
or formulas) occurring on the same level such that in the compressed \emph{%
dag-like} proof every level will contain mutually different objects. In the
case of SC the compressed polynomial-height dag-like proofs still might be
too large due to possibly exponential number of distinct sequents occurring
in it. Therefore we prefer to compress tree-like ND proofs whose height is
polynomial in the size of conclusion. Since ND proofs operate with single
formulas (not sequents!), the horizontal compression in question will
provide us with desired polynomial-weight dag-like deductions. However, our
dag-like proofs require a more sophisticated notion of provability.

\textbf{Basic example} \newline
Our basic ND calculus for minimal logic, \textsc{NM}$_{\rightarrow }$,
includes two inferences $\left( \rightarrow I\right) $ and $\left(
\rightarrow E\right) $ ($\rightarrow $-\emph{introduction} and $\rightarrow $%
-\emph{elimination }\footnote{{\small {\footnotesize also known as \emph{%
modus ponens}.}}}): 
\begin{equation*}
\fbox{$\left( \rightarrow I\right) :\dfrac{\beta }{\alpha \rightarrow \beta }
$}\ \ and\ \ \fbox{$\ \left( \rightarrow E\right) \,:\dfrac{\alpha \quad
\quad \alpha \rightarrow \beta }{\beta }$}.
\end{equation*}
Its extension, \textsc{NM}$_{\rightarrow }^{+}$, also includes \emph{%
multipremise repetitions} $\left( R\right) _{n}$, $n\geq 1$: 
\begin{equation*}
\fbox{$\ \left( R\right) _{n}:\overset{n\ times}{\dfrac{\overbrace{\alpha
\quad \cdots \quad \alpha }}{\alpha }}\ $},
\end{equation*}
for $\alpha $, $\beta $, $\gamma $, etc. ranging over purely implicational
formulas of minimal logic.

A maximal deductive path $\Theta \!=\!\left[ x_{0},\cdots ,x_{h}=r\right] $
connecting leaf $x_{0}$ labeled $\alpha $ with the root $r$ labeled $\rho $
is called \emph{closed} (otherwise \emph{open}) if there exists $i<h$ such
that $x_{i}$ is the conclusion of $\left( \rightarrow I\right) $ labeled $%
\alpha \rightarrow \beta $ (for some $\beta $); these \emph{discharged}
assumptions are decorated as $\left[ \alpha \right] $ (possibly indexed like 
$\left[ \alpha \right] ^{i}$ for superscripts $^{(i)}$ decorating $\alpha
\rightarrow \beta $ involved). Now $\rho $ is called \emph{provable} in 
\textsc{NM}$_{\rightarrow }$ (resp. \textsc{NM}$_{\rightarrow }^{+}$) by a
given tree-like deduction $\partial $ that in turn is called a \emph{%
tree-like proof }of $\rho $ (abbr. $\partial \vdash \rho $), iff every
maximal deductive path in $\partial $ is closed.

Consider a tree-like deduction $\partial \in $\textsc{NM}$_{\rightarrow }$
of $\tau \rightarrow \sigma $:$\smallskip $ $\smallskip $ $\smallskip $

{\small \hspace{-20pt} 
\begin{tikzpicture}
  \draw (0.6,2.5) node {
    \AxiomC{$[\alpha]^{1}$}
    \AxiomC{$[\alpha \rightarrow\beta]^{2}$}
    \BinaryInfC{$\beta$}
    \UnaryInfC{$\gamma \rightarrow \beta$}    
    \UnaryInfC{$ (\alpha\rightarrow \beta) \rightarrow(\gamma \rightarrow \beta)^{(2)}$}
    \UnaryInfC{$\alpha\rightarrow ((\alpha\rightarrow \beta) \rightarrow(\gamma \rightarrow \beta))^{(1)}$}
    \UnaryInfC{$\alpha\rightarrow ((\alpha\rightarrow \beta) \rightarrow(\gamma \rightarrow \beta))$}
    \AxiomC{$[\delta]^{3}$}
    \AxiomC{$[\delta\rightarrow\beta]^{4}$}
    \BinaryInfC{$\beta$}
    \UnaryInfC{$\gamma \rightarrow \beta$}  
    \UnaryInfC{$(\delta\rightarrow \beta) \rightarrow(\gamma \rightarrow \beta)^{(4)}$}
    \UnaryInfC{$\delta \rightarrow ((\delta\rightarrow \beta) \rightarrow(\gamma \rightarrow \beta))^{(3)}$}
    \AxiomC{$[\tau]^{5}$}
    \UnaryInfC{$\left.\overset{} \tau \right.$}
    \UnaryInfC{$\left.\overset{}\tau\right.$}
    \UnaryInfC{$\left. \overset{}\tau \right.$}
    \UnaryInfC{$\left.\overset{} \tau \right.$}
   \BinaryInfC{$(\alpha\rightarrow ((\alpha \rightarrow \beta) \rightarrow (\gamma\rightarrow\beta))) \rightarrow \sigma$}
    \BinaryInfC{$\sigma$}
    \UnaryInfC{$\tau \rightarrow \sigma^{(5)}$}
   \UnaryInfC{$$}
    \DisplayProof
  } ;
\end{tikzpicture}
\newline
where \newline
$\tau :=(\delta \rightarrow ((\delta \rightarrow \beta )\rightarrow (\gamma
\rightarrow \beta )))\rightarrow ((\alpha \rightarrow ((\alpha \rightarrow
\beta )\rightarrow (\gamma \rightarrow \beta )))\rightarrow \sigma )$. 
\newline
} Obviously every maximal deductive path in $\partial $ is closed, and hence 
$\tau \rightarrow \sigma $ is provable in \textsc{NM}$_{\rightarrow }$.

Furthermore consider a (horizontally) compressed dag-like deduction $%
\partial ^{\prime }\in $\textsc{NM}$_{\rightarrow }^{+}$ that is obtained by
merging two occurrences $\gamma \rightarrow \beta $ using repetition $\left(
R\right) _{2}$ with conclusion $\beta $ : \footnote{{\small {\footnotesize %
Obviously this operation destroys tree structure of }$\partial $%
{\footnotesize . }}}\newline

{\small \hspace{-20pt} 
\begin{tikzpicture}
  \draw (0,2.5) node {
    \AxiomC{$[\alpha]^{1}$}
    \AxiomC{$[\alpha \rightarrow\beta]^{2}$}
    \BinaryInfC{$\beta $}
    \AxiomC{$[\delta]^{3}$}
    \AxiomC{$[\delta\rightarrow\beta]^{4}$}
    \BinaryInfC{$\beta$} 
    \BinaryInfC{$\beta$}
    \UnaryInfC{$\gamma \rightarrow \beta$} 
    \UnaryInfC{$$} 
       \DisplayProof
  } ;
  \draw [->] (0,1.5) -- (2.5,1) ;
  \draw [->] (0,1.5) -- (-2.5,1) ;
  \draw (1,0) node {
      \AxiomC{}
    \UnaryInfC{\!$ (\alpha\!\rightarrow \!\beta) \!\rightarrow\!(\gamma \!\rightarrow \!\beta)^{(2)}$}
    \UnaryInfC{\!$\alpha\!\rightarrow \!((\alpha\!\rightarrow \!\beta)\! \rightarrow\!(\gamma\! \rightarrow \!\beta))^{(1)}$}
    \UnaryInfC{$\alpha\!\rightarrow \!((\alpha\rightarrow\! \beta)\! \rightarrow\!(\gamma\! \rightarrow\! \beta))$}
    \AxiomC{}
    \UnaryInfC{\!\!$(\delta\!\rightarrow\! \beta) \!\rightarrow\!(\gamma \!\rightarrow \!\beta)^{(4)}$}
   \UnaryInfC{$\!\!\delta\! \rightarrow\! ((\delta\!\rightarrow\! \beta)\! \rightarrow\!(\gamma \!\rightarrow \!\beta))^{(3)}$}
    \AxiomC{$\QDATOP{\left[ \tau \right]^{5} }{ \,\downarrow }$}
    \UnaryInfC{$\left. \overset{}\tau \right.$}
    \UnaryInfC{$\left. \overset{%
}\tau \right.$}
    \BinaryInfC{$(\alpha\!\rightarrow\! ((\alpha\! \rightarrow\! \beta)\! \rightarrow\! (\gamma\!\rightarrow\!\beta))) \!\rightarrow \!\sigma$}
    \BinaryInfC{$\sigma$}
    \UnaryInfC{$\tau \!\rightarrow \!\sigma^{(5)}$} 
   \UnaryInfC{$$}
   
 \DisplayProof
  } ;
\end{tikzpicture}
\newline
} Note that $\partial ^{\prime }$ has open deductive paths like those
connecting assumptions $\alpha $ (or $\alpha \rightarrow \beta $) with $\tau
\rightarrow \sigma $ via $(\delta \rightarrow \beta )\rightarrow (\gamma
\rightarrow \beta )$ and/or $\delta $ (or $\delta \rightarrow \beta $) via $%
(\alpha \rightarrow \beta )\rightarrow (\gamma \rightarrow \beta )$.
However, the closure of other (called \emph{regular}) maximal deductive
paths in $\partial ^{\prime }$ already confirms the provability of $\tau
\rightarrow \sigma $. This is readily seen by tree-like unfolding of $%
\partial ^{\prime }$ back to $\partial $ that eliminates all open maximal
deductive paths in $\partial ^{\prime }$. Hence ``naive'' tree-like
assertion \emph{`every maximal deductive path is closed'} is \textbf{not a
necessary condition} for the dag-like provability. Instead, we replace it by
a more liberal \textbf{sufficient} dag-like condition \emph{`every regular
deductive path is closed'}, where, in $\partial ^{\prime }$, the regular
deductive paths are naturally determined by the repetition rule $\left(
R\right) _{2}$ involved. (In the general case of arbitrary dag-like proofs
we'll extend \textsc{NM}$_{\rightarrow }^{+}$ by adding appropriate
conditions of regularity that enable us to formalize this notion in
polynomial way).

To complete discussion we observe that $\partial $ and $\partial ^{\prime }$
both have the same total number of nodes, due to minimum size of the
multipremise repetition involved. This might downplay crucial importance of
our horizontal compression that in the general case provides compressed
dag-like deductions which tends to be exponentially smaller than tree-like
origins.

\section{Survey of proofs}

Recall that the minimal logic \cite{Joh} is determined by two axioms $\alpha
\rightarrow \left( \beta \rightarrow \alpha \right) $, $\left( \alpha
\rightarrow \left( \beta \rightarrow \gamma \right) \right) \rightarrow
\left( \left( \alpha \rightarrow \beta \right) \rightarrow \left( \alpha
\rightarrow \gamma \right) \right) $\ and rule \emph{modus ponens} $\fbox{$%
\dfrac{\alpha \quad \quad \alpha \rightarrow \beta }{\beta }$}$ in the
Hilbert-style formalism whose vocabulary includes propositional variables
and propositional connective `$\rightarrow $' ($\alpha $, $\beta $, $\gamma $%
, etc. denote corresponding formulas).

\subsection{Basic tree-like Natural Deductions}

{Our basic Natural Deduction calculus for minimal logic, \textsc{NM}$%
_{\rightarrow }$, includes two kinds of inference rules $\left( \rightarrow
I\right) $ and $\left( \rightarrow E\right) $: 
\begin{equation*}
\fbox{$\left( \rightarrow I\right) :\dfrac{\beta }{\alpha \rightarrow \beta }
$}\ \ and\ \ \fbox{$\ \left( \rightarrow E\right) \,:\dfrac{\alpha \quad
\quad \alpha \rightarrow \beta }{\beta }$}.
\end{equation*}
Its extended version, \textsc{NM}$_{\rightarrow }^{+}$, also includes
repetition rules $\left( R\right) _{n}$, $n\geq 1$: 
\begin{equation*}
\fbox{$\ \left( R\right) _{n}:\overset{n\ times}{\dfrac{\overbrace{\alpha
\quad \cdots \quad \alpha }}{\alpha }}\ $}.
\end{equation*}
Furthermore, denote by \textsc{NM}$_{\rightarrow }^{\text{\textsc{t}}}$ a
subsystem of \textsc{NM}$_{\rightarrow }^{+}$ that includes only special
instances of $\left( R\right) _{n}, n>1$, whose identical premises are
conclusions of pairwise different instances of $\left( \rightarrow I\right) $%
, $\left( \rightarrow E\right) $ and/or $\left( R\right) _{1}$ (which we
also abbreviate $\left( R\right) $). Such instances of multipremise
repetitions we call \emph{separation rules}, abbr.: $\left( S\right) $. So $%
\left( S\right) $ can be a $\left( R\right) _{2}$ from Basic example
(above): 
\begin{equation*}
\fbox{$\!\dfrac{\dfrac{\alpha \ \ \ \alpha \rightarrow \beta }{\beta }\quad
\quad \dfrac{\delta \ \ \ \delta \rightarrow \beta }{\beta }}{\beta }$}\text{%
,}
\end{equation*}
while mostly involved $\left( S\right) $ is $\left( R\right) _{s+2}$ of the
following shape, for pairwise different $\alpha, \beta$ and $\gamma
_{1},\cdots, \gamma _{s}$: 
\begin{equation*}
\fbox{$\!\dfrac{\dfrac{\beta }{\alpha \rightarrow \beta }\quad \dfrac{\alpha
\rightarrow \beta }{\alpha \rightarrow \beta }\quad \dfrac{\gamma _{1}\ \ \
\gamma _{1}\!\rightarrow \!\left( \alpha \!\rightarrow \!\beta \right) }{%
\alpha \rightarrow \beta }\quad \cdots \ \quad \dfrac{\gamma _{s}\ \quad
\gamma _{s}\!\rightarrow \!\left( \alpha \!\rightarrow \!\beta \right) }{%
\alpha \rightarrow \beta }}{\alpha \rightarrow \beta }$}\text{.}
\end{equation*}
Tree-like \textsc{NM}$_{\rightarrow }^{(-)}$-deductions $\partial $ are
finite rooted trees whose nodes are labeled with purely implicational
formulas ($\alpha ,\beta ,\gamma $, etc.) ordered in accordance with the
inferences exposed; we also assume that all top nodes (the leaves) have the
same height $h$ in $\partial $ (this is easily achieved using $\left(
R\right) $, if necessary). }

In any given $\partial $, a deductive path $\Theta \!=\!\left[ x_{0},\cdots
,x_{h}=r\right] $ connecting leaf $x_{0}$ labeled $\alpha $ with the root $r$
labeled $\rho $ is called \emph{closed} if there exists a $i<h$ such that $%
x_{i}$ is the conclusion of $\left( \rightarrow I\right) $ labeled $\alpha
\rightarrow \beta $ (for some $\beta $) -- such $\left( \rightarrow I\right) 
$ we also denote $\left( \rightarrow _{\alpha }I\right) $ and call \emph{%
discharged} corresponding assumptions $\alpha $ (which for the sake of
brevity are exposed in square brackets).

\begin{definition}
A given tree-like \textsc{NM}$_{\rightarrow }$(\textsc{NM}$_{\rightarrow
}^{+}$, \textsc{NM}$_{\rightarrow }^{\text{\textsc{t}}}$)-deduction $%
\partial $ \emph{proves} the conclusion $\rho $ (abbr.: $\partial \vdash
\rho $) iff its every deductive path $\Theta $ is closed; such $\partial $
is called a \emph{tree-like proof} of $\rho $. A purely implicational
formula $\rho $ is \emph{valid in minimal logic} iff there exists a
tree-like proof of $\rho $ in \textsc{NM}$_{\rightarrow }$.
\end{definition}

\begin{lemma}
For any purely implicational formula $\rho $, the following conditions (A),
(B), (C) and (D) are equivalent.

(A) $\rho $ is valid in minimal logic.

(B) There exists a tree-like proof of $\rho $ in \textsc{NM}$_{\rightarrow
}^{+}$.

(C) There exists a tree-like proof of $\rho $ in \textsc{NM}$_{\rightarrow
}^{\text{\textsc{t}}}$.

(D) $\rho $ is derivable in Hilbert-style formalism of minimal logic (see
above).
\end{lemma}

\begin{proof}
(A) $\equiv $ (D) is well-known (see e.g. \cite{Prawitz}) and proved by
induction on the heights of deductions involved. (B) $\equiv $ (C) $\equiv $
(D) easily follows by the same token.
\end{proof}

\subsection{Basic dag-like ND}

Now our basic \emph{dag-like} calculus for minimal logic, \textsc{NM}$%
_{\rightarrow }^{\text{\textsc{d}}}$, is a dag-like version of \textsc{NM}$%
_{\rightarrow }^{\text{\textsc{t}}}$, except that $\left( S\right) $ are
being understood disjunctively: ``\emph{if at least one premise is proved
then so is the conclusion}'' (in contrast to standard conjunctive
interpretation: ``\emph{if all premises are proved then so is the conclusion}%
''). This requires a more sophisticated definition of the \textsc{NM}$%
_{\rightarrow }^{\text{\textsc{d}}}$-provability, as follows.

Arbitrary \textsc{NM}$_{\rightarrow }^{\text{\textsc{d}}}$-deductions $%
\partial $ are finite rooted dags (directed acyclic graphs) whose nodes $x$
are labeled with formulas $\ell \left( x\right) $. Let $V=V\left( \partial
\right) $ and $E=E\left( \partial \right) $ denote respectively the sets of
nodes and directed edges $e=\left\langle x,y\right\rangle $ whose sources $x$
(also called \emph{parents}) occur above targets $y$ (also called \emph{%
children}). Note that unlike in \textsc{NM}$_{\rightarrow }^{\text{\textsc{t}%
}}$, parents in \textsc{NM}$_{\rightarrow }^{\text{\textsc{d}}}$ can have
many children.

By $r\in V$ and $\rho =\ell \left( r\right) $ we denote respectively the
root and conclusion of $\partial $. For any $x\in V$ let $h\left( x\right) $
denote the height of $x$ in $\partial $ where $h(r)=0$. $A$ will denote the
set of top formulas (also called \emph{assumptions}) occurring in $\partial $%
. We also assume that all top nodes (the leaves) have the same height $%
h=\max \left\{ h\left( x\right) :x\in V\right\} $. Let $\left\lceil
x\right\rceil $ abbreviate names $\left( \rightarrow I\right) $, $\left(
\rightarrow E\right) $, $\left( R\right) $, $\left( S\right) $ of inferences
with conclusion $x$ while assuming that $\left\lceil r\right\rceil \neq
\left( R\right) _{n}$.

For any $x,y\in V$ with $h\left( x\right) <h\left( y\right) $, $x\prec y$
will abbreviate ``$x$ is \emph{reachable} from $y$ (: there is a deductive
path from $y$ to $x$)''; $e_{1}\prec e_{2}$ will denote $x_{1}\preceq y_{2}$
for $e_{i}=\left\langle x_{i},y_{i}\right\rangle $ ($i\in \left\{
1,2\right\} $), where as usual $x\preceq y$ stands for $x\prec y\vee x=y$.

For any $e=\left\langle x,y\right\rangle \in E$ let $h\left( e\right)
:=h\left( x\right) $. For any $e=\left\langle x,y\right\rangle \in E $\
define the assumptions $A\left( e\right) =A\left( x\right) \subseteq A$ by $%
A\left( x\right) :=\left\{ \ell \left( y\right) \in A:x\preceq y\,\&h\left(
y\right) =h\right\} $. Thus $A\left( r\right) =A$.

For any $x\in V$ we denote by $IV\left( x\right) $ and/or $IE\left( x\right) 
$ and $OV\left( x\right) $ and/or $OE\left( x\right) $ the sets of $x$%
-ingoing and $x$-outgoing vertices and/or edges, respectively.

For any $e=\left\langle x,y\right\rangle \in E$ and $\alpha \in A$\ we also
let $IE\left( e\right) :=IE\left( x\right) $, $OE\left( e\right) :=OE\left(
y\right) $, $IE\left( e,\alpha \right) =IE\left( x,\alpha \right) :=\left\{
\left\langle z,x\right\rangle \in IE\left( x\right) :\alpha \in A\left(
x\right) \right\} $ and $OE\left( e,\alpha \right) \!=\!OE\left( y,\alpha
\right) \!:=\!\left\{ \left\langle y,z\right\rangle \in OE\left( y\right)
:\alpha \in A\left( y\right) \right\} $. Thus for any $x$, $A\left( x\right)
=\underset{e\in I\!E\left( x\right) }{\bigcup }A\left( e\right) $.

We assume that every deduction $\partial $ in question satisfies standard
conditions of local correctness with respect to the inferences involved
together with auxiliary assumption $\left\lceil r\right\rceil \in \left\{
\left( \rightarrow I\right) ,\left( \rightarrow E\right) \right\} $.

Recall that $|OE\left( x\right) |>1$ is possible in \textsc{NM}$%
_{\rightarrow }^{\text{\textsc{d}}}$ but not in \textsc{NM}$_{\rightarrow }^{%
\text{\textsc{t}}}$.\footnote{{\small {\footnotesize Dag-like deduction }$%
\partial ^{\prime }${\footnotesize \ from Basic example (above) is not a
tree-like one as the node }$x${\footnotesize \ labeled }$\gamma \rightarrow
\beta $ {\footnotesize has two children, i.e. }$\left| OE\left( x\right)
\right| =2>1 ${\footnotesize . }}}

To formalize the dag-like provability we'll also assume that every $%
\partial\in\QTR{sc}{NM}_{\rightarrow }^{\text{\textsc{d}}}$ is enriched with
a special \emph{flow function} $f\!:E\times E\times A\!\rightarrow \!\wp
\left( E\right) $ such that for any $e_{1}\prec e_{2}$ and $\alpha \in
A\left( e_{2}\right) $ in $\partial $ it holds $f\!\left( e_{1},e_{2},\alpha
\right) \subseteq IE\left( e_{2},\alpha \right) $. \footnote{{\small 
{\footnotesize Intuitively, any $e\in f\left( e_{1},e_{2},\alpha \right) $
should indicate that there exists ``regular'' deductive path (also called $f$%
-thread) from assumpion $\alpha $ to the root via edges $e$, $e_{2}$ and $%
e_{1}$\ (cf. Defintion 4).}}} The enriched deduction we'll denote $\partial
_{f}$, rather than $\partial $.

\begin{definition}
In $\partial _{f}\in \!\QTR{sc}{NM}_{\rightarrow }^{\text{\textsc{d}}}$, $f$%
\emph{-threads} are maximal deductive paths\newline
$\Theta \!=\!\left[ x_{0},\!\cdots \!,x_{h}\!=\!r\right] $ from leaves $%
x_{0} $ to the root $r$ such that 
\begin{equation*}
\left( \forall 0<i<j<h\right) \left\langle x_{i-1},x_{i}\right\rangle \!\in
\!f\!\left( \left\langle x_{j},x_{j+1}\right\rangle ,\left\langle
x_{i},x_{i+1}\right\rangle ,\ell \left( x_{0}\right) \right) .
\end{equation*}
A\negthinspace\ $f$-thread $\Theta $ is \emph{closed} if it contains $\left(
\rightarrow _{\ell \left( x_{0}\right) }\!I\right) $,\negthinspace\
i.e.\negthinspace\ $\left( \exists i<h\right) \!\left\lceil
x_{i}\right\rceil \!=\!\left( \rightarrow _{\ell \left( x_{0}\right)
}I\right) $, \emph{open} otherwise. A given deduction $\partial _{f}\in \! 
\QTR{sc}{NM}_{\rightarrow }^{\text{\textsc{d}}}$ \emph{proves} its
conclusion $\rho =\ell \left( r\right) $ iff every $f$-thread in $\partial
_{f}$ is closed; such $\partial _{f}$ is called a \emph{proof} of $\rho $
(abbr.: $\partial _{f}\vdash \rho $). \footnote{{\small {\footnotesize This
criterion of dag-like provability is weaker than naive requirement ``every
deductive path is closed'' that might be suggested by the tree-like case
(cf. Definition 1).}}}
\end{definition}

\begin{definition}
A $\partial _{f}\in \!\QTR{sc}{NM}_{\rightarrow }^{\text{\textsc{d}}} $ is 
\emph{regular} iff $f$ satisfies the following conditions 1--3, for all $%
x_{i}\in V$, $e_{i}=\left\langle x_{i},y_{i}\right\rangle \in E$ with $%
e_{1}\prec e_{2}$\ and $\alpha \in A\left( e_{2}\right) $, in $\partial$.

\begin{enumerate}
\item  $\emptyset \neq f\!\left( e_{1},e_{2},\alpha \right) \subseteq
IE\left( e_{2},\alpha \right) $.

\item  If $\left\lceil x_{2}\right\rceil \!=\!\left( \rightarrow \!E\right) $
then $\!f\!\left( e_{1},e_{2},\alpha \right) =IE\left( e_{2},\alpha \right) $%
.

\item  If $e_{0}\prec e_{1}$ then $f\!\left( e_{1},e_{2},\alpha \right)
\subseteq f\!\left( e_{0},e_{2},\alpha \right) $.
\end{enumerate}

Any $f$-thread in a given regular proof $\partial _{f}$ is called a \emph{%
regular deductive path}.
\end{definition}

\begin{theorem}
Any given $\rho $ is valid in minimal logic iff there exists a regular proof 
$\partial _{f}\in \!\QTR{sc}{NM}_{\rightarrow }^{\text{\textsc{d}}}$ of $%
\rho $.
\end{theorem}

\begin{proof}
1. The existence of a regular proof in question easily follows from
Definition 1, since any tree-like proof in \textsc{NM}$_{\rightarrow }$ is a
dag-like proof in \textsc{NM}$_{\rightarrow }^{\text{\textsc{d}}}$ with
respect to trivial flow function $f\left( e_{1},e_{2},\alpha \right)
:=IE\left( e_{2},\alpha \right) $. \newline
2. To show the inverse we unfold regular $\partial _{f}\in \,$\textsc{NM}$%
_{\rightarrow }^{\text{\textsc{d}}}$ according to ``roadmap'' $f$\ into a
tree-like deduction $\partial ^{\prime }=\partial _{f}^{\prime }\left(
h\right) \in \,$\textsc{NM}$_{\rightarrow }^{\text{\textsc{t}}}$ of $\rho $,
where for all $k\leq h$, intermediate deductions $\partial _{f}^{\prime
}\left( k\right) $ arise by recursion on $k$ such that $\partial
_{f}^{\prime }\left( k\right) $ is tree-like, i.e. $\left| OE\left( x\right)
\right| \leq 1$, for all $x\in V$ of the height $\leq k$. Below for any
dag-like deduction $\partial $ involved and any $x\in V$ we denote by $%
\left( \partial \right) _{x}$ the subdag rooted in $x$: 
\begin{equation*}
\left( \partial \right) _{x}=\dfrac{\left[ \left( \partial \right) _{y}:y\in
IV\left( x\right) \right] }{\ell \left( x\right) }
\end{equation*}

\emph{Basis of recursion}. $\partial _{f}^{\prime }\left( 0\right) =\partial
_{f}^{\prime }\left( 1\right) :=\partial _{f}$.

\emph{Recursion step}. $\partial _{f}^{\prime }\!\left( k\!+\!1\right) $
arises as follows by replacing subdags $\left( \!\partial _{f}^{\prime
}\!\left( k\right) \right) \!_{x}$ by new subdags $\left( \partial
_{f}^{\prime }\!\left( k\!+\!1\right) \right) \!_{x}$\negthinspace , or
collections $\left( \partial _{f}^{\prime }\left( k\!+\!1\right) \right)
\!_{x_{1}}\!,\cdots ,\left( \partial _{f}^{\prime }\left( k\!+\!1\right)
\right) \!_{x_{m}}$, for all nodes $x$ of the height $k\!+\!1$ in $\partial
_{f}^{\prime }\left( k\right) $. Consider two cases.

2.1. Suppose $\left| OE\left( x\right) \right| \leq 1$. Then let $\left(
\partial _{f}^{\prime }\left( k+1\right) \right) \!_{x}:=\left( \partial
_{f}^{\prime }\left( k\right) \right) \!_{x}$.

2.2. Suppose $\left| O\!E\left( x\right) \right| =m>1$. Consider distinct
conclusions $y_{1},\cdots ,y_{m}$ of the edges from $OE\left( x\right) $, in 
$\partial _{f}^{\prime }\left( k\right) $. Let $x_{1},\cdots ,x_{m}$\ be
identically labeled $\ell \left( x_{i}\right) :=\ell \left( x\right) $ nodes
connected with $y_{1},\cdots ,y_{m}$ by new edges $e_{1}\!:=\!\left\langle
x_{1},y_{1}\right\rangle ,\cdots \!,e_{m}\!:=\left\langle
x_{m},y_{m}\right\rangle $ supplied with the old inferences $\left\lceil
y_{1}\right\rceil ,\cdots ,\left\lceil y_{m}\right\rceil $. The incoming
edges $I\!E\left( y_{i}\right) $, $i\in \left[ m\right] $, however, might
differ from the old ones, if $\left\lceil y_{i}\right\rceil =\left( S\right) 
$. Namely, if $\left\lceil y_{i}\right\rceil =\left( S\right) $ then the new 
$I\!E\left( y_{i}\right) $ should include only the old edges from $\underset{%
\partial _{f}^{\prime }\left( k\right) \ni e^{\prime }\preceq e_{i}}{\bigcap 
}f\left( e^{\prime },e_{i},\alpha \right) $, for any $\alpha \in A\left(
e_{i}\right) $, (recall that $\left\{ e^{\prime }\in \partial _{f}^{\prime
}\left( k\right) :e^{\prime }\preceq e_{i}\right\} $ are linearly ordered).
In the cases $\left\lceil y_{i}\right\rceil \in \left\{ \left( \rightarrow
\!I\!\right) ,\left( \rightarrow E\right) ,\left( R\right) \right\} $, $%
I\!E\left( y_{i}\right) $ remains unchanged. Having done this for all $i\in %
\left[ m\right] $ we replace $\left( \!\partial _{f}^{\prime }\!\left(
k\right) \right) \!_{x}$ by the entire collection $\left( \!\partial
_{f}^{\prime }\!\left( k\!+\!1\right) \right) \!_{x_{1}},\cdots \!,\left(
\!\partial _{f}^{\prime }\!\left( k\!+\!1\right) \right) \!_{x_{m}}$, where
every $\left( \partial _{f}^{\prime }\left( k+1\right) \right) \!_{x_{i}}$
arises by the corresponding upgrade of $\left( \partial _{f}^{\prime }\left(
k\right) \right) \!_{x}$.

This completes the recursive definition of $\partial _{f}^{\prime }\left(
k\right) $, $k\leq h$, which yields the desired unfolded tree-like deduction 
$\partial ^{\prime }=\partial _{f}^{\prime }\left( h\right) \in \,$\textsc{NM%
}$_{\rightarrow }^{\text{\textsc{t}}}$ of $\rho $. Since $\partial _{f}$ is
regular, it is easily verified that all deductive paths in $\partial
_{f}^{\prime }\left( h\right) $ are closed, as they coincide with closed $f$%
-threads occurring in the underlying proof $\partial _{f}\in \,$\textsc{NM}$%
_{\rightarrow }^{\text{\textsc{d}}}$. So $\partial ^{\prime }\vdash \rho $
in $\,$\textsc{NM}$_{\rightarrow }^{\text{\textsc{t}}}\subset \,\text{%
\textsc{NM}}_{\rightarrow }^{+}$ and hence \negthinspace $\rho $ is valid by
Lemma 2.
\end{proof}

In any given regular $\partial _{f}\in \QTR{sc}{NM}_{\rightarrow }^{\text{%
\textsc{d}}}$ we assign edges $e=\left\langle x,y\right\rangle \in
O\!E\left( x\right) $ with \emph{proper assumptions} $A_{f}(e)\!\subseteq
\!A(e)$ that are defined in two steps, as follows. First recall a natural
recursive definition of $A\left( e\right) $:

\begin{enumerate}
\item  If $h(e)=h$ then $A(e):=\left\{ \ell \left( x\right) \right\} $.

\item  If $h(e)<h$ then $A(e):=\underset{e^{\prime }\in I\!E\left( x\right) 
}{\bigcup }A\left( e^{\prime }\right) $.
\end{enumerate}

Now $A_{f}(e)\!$ is defined by successive thinning of $A(e)$ recursively as
follows, where $I\!\left( e\right) :=\left\{ \alpha :\!\left\lceil
y\right\rceil =\left( \rightarrow _{\alpha }I\right) \right\} $ and $%
\overline{f}\left( e,e^{\prime },\alpha \right) :=IE\left( e^{\prime
},\alpha \right) \setminus f\left( e,e^{\prime },\alpha \right) $.

\begin{enumerate}
\item  If $h(e)=h$ then $A_{f}(e):=A(e)=\left\{ \ell \left( x\right)
\right\} $.

\item  If $h(e)<h$ then $A_{f}(e):=A(e)\setminus \left( \!I\!\left( e\right)
\cup \underset{\alpha \in A(e)\&e\prec e^{\prime }\&e^{\prime \prime }\in 
\overline{f}\left( e,e^{\prime },\alpha \right) }{\bigcup }A_{f}\left(
e^{\prime \prime }\right) \right) $.
\end{enumerate}

Also let $A_{f}(r):=\underset{e\in I\!E\left( r\right) }{\bigcup }%
A_{f}\left( e\right) $.

\begin{lemma}
For any regular $\partial _{f}\in \!\QTR{sc}{NM}_{\rightarrow }^{\text{%
\textsc{d}}}$ it holds $\partial _{f}\vdash \ell \left( r\right)
\Leftrightarrow A_{f}(r)=\emptyset $.
\end{lemma}

\begin{proof}
For any $f$-thread $\Theta =\left[ x_{0},\cdots ,x_{h}=r\right] $ and $%
e=\left\langle x_{i},x_{i+1}\right\rangle $ for $i<h$ (abbr.: $e\in \Theta )$
let $\ell \left( \Theta \right) :=\ell \left( x_{0}\right) $ and $\Theta
\!\upharpoonright _{e}:=\left[ x_{0},\cdots ,x_{i+1}\right] \!\subseteq
\Theta $. Let $O\left( \Theta \!\upharpoonright _{e}\right) $ abbreviate
that $\Theta \!\!\upharpoonright _{e}$ is open, i.e. $\left( \forall j\leq
i\right) \left\lceil x_{j}\right\rceil \!\neq \!\left( \rightarrow _{\ell
\left( x_{0}\right) }I\right) $, and let $O\left( \Theta \right) $
abbreviate that $\Theta $ is open, i.e. $\left( \forall j<h\right)
\left\lceil x_{j}\right\rceil \!\neq \!\left( \rightarrow _{\ell \left(
x_{0}\right) }I\right) $. Furthermore, for any $\alpha \in A$ denote by $%
\mathcal{T}_{f}\left( \alpha \right) $ the set of $f$-threads $\Theta $ with 
$\alpha =\ell \left( \Theta \right) $. Then for any $e\in E$ we obtain 
\begin{equation*}
A_{f}(e)=\left\{ \alpha \in A\left( e\right) :\left( \exists \,\Theta \in 
\mathcal{T}_{f}\left( \alpha \right) \right) \left( e\in \Theta
\,\&\,O\!\left( \Theta \!\upharpoonright _{e}\right) \right) \right\} .
\end{equation*}
This follows from Definition 4 by straightforward induction on $h-h\left(
e\right) $. Now 
\begin{equation*}
\begin{array}{l}
A_{f}(r)=\emptyset \Leftrightarrow \left( \forall e\in I\!E\left( r\right)
\right) A_{f}\!\left( e\right) =\emptyset \\ 
\\ 
\Leftrightarrow \left( \forall e\in I\!E\left( r\right) \right) \left(
\forall \alpha \in A\right) \left( \forall \,\Theta \in \mathcal{T}%
_{f}\left( \alpha \right) \right) \left( e\in \Theta \Rightarrow \lnot
O\!\left( \Theta \!\upharpoonright _{e}\right) \right) \\ 
\\ 
\Leftrightarrow \left( \forall \alpha \in A\right) \left( \forall \,\Theta
\in \mathcal{T}_{f}\left( \alpha \right) \right) \lnot O\!\left( \Theta
\right)
\end{array}
\end{equation*}
and hence $A_{f}(r)=\emptyset $ iff every $f$-thread in $\partial _{f}$ is
closed, i.e. $\partial _{f}\vdash \ell \left( r\right) $.
\end{proof}

Together with Theorem 5 this yields

\begin{corollary}
A purely implicational formula $\rho $ is valid in minimal logic iff there
exists a regular dag-like deduction $\partial _{f}\in \!\QTR{sc}{NM}%
_{\rightarrow }^ {\text{\textsc{d}}}$ of $\rho =\ell \left( r\right) $ such
that $A_{f}(r)=\emptyset $.
\end{corollary}

\begin{example}
\footnote{{\footnotesize This example was briefly discussed in \S 1.2. For
brevity here we drop superscripts $^{(i)}$.}} Consider a tree-like proof $%
\partial \in \,$\textsc{NM}$_{\rightarrow }$ of $\tau \rightarrow \sigma $:$%
\smallskip $ $\smallskip $ $\smallskip $

{\small \hspace{-20pt} 
\begin{tikzpicture}
  \draw (0.6,2.5) node {
    \AxiomC{$[\alpha]^{1}$}
    \AxiomC{$[\alpha \rightarrow\beta]^{2}$}
    \BinaryInfC{$\beta$}
    \UnaryInfC{$\gamma \rightarrow \beta$}    
    \UnaryInfC{$ (\alpha\rightarrow \beta) \rightarrow(\gamma \rightarrow \beta)$}
    \UnaryInfC{$\alpha\rightarrow ((\alpha\rightarrow \beta) \rightarrow(\gamma \rightarrow \beta))$}
    \UnaryInfC{$\alpha\rightarrow ((\alpha\rightarrow \beta) \rightarrow(\gamma \rightarrow \beta))$}
    \AxiomC{$[\delta]^{3}$}
    \AxiomC{$[\delta\rightarrow\beta]^{4}$}
    \BinaryInfC{$\beta$}
    \UnaryInfC{$\gamma \rightarrow \beta$}  
    \UnaryInfC{$(\delta\rightarrow \beta) \rightarrow(\gamma \rightarrow \beta))$}
    \UnaryInfC{$\delta \rightarrow ((\delta\rightarrow \beta) \rightarrow(\gamma \rightarrow \beta))$}
    \AxiomC{$[\tau]^{5}$}
    \UnaryInfC{$\left.\overset{} \tau \right.$}
    \UnaryInfC{$\left.\overset{}\tau\right.$}
    \UnaryInfC{$\left. \overset{}\tau \right.$}
    \UnaryInfC{$\left.\overset{} \tau \right.$}
   \BinaryInfC{$(\alpha\rightarrow ((\alpha \rightarrow \beta) \rightarrow (\gamma\rightarrow\beta))) \rightarrow \sigma$}
    \BinaryInfC{$\sigma$}
    \UnaryInfC{$\tau \rightarrow \sigma$}
   \UnaryInfC{$$}
    \DisplayProof
  } ;
\end{tikzpicture}
\newline
where \newline
$\tau :=(\delta \rightarrow ((\delta \rightarrow \beta )\rightarrow (\gamma
\rightarrow \beta )))\rightarrow ((\alpha \rightarrow ((\alpha \rightarrow
\beta )\rightarrow (\gamma \rightarrow \beta )))\rightarrow \sigma )$. 
\newline
} All deductive paths in $\partial $ are closed and hence \smallskip $\tau
\rightarrow \sigma $ is provable in \textsc{NM}$_{\rightarrow }$.\newline
Define a compressed dag-like deduction $\partial ^{\prime }\in \QTR{sc}{NM}%
_{\rightarrow }^{\text{\textsc{d}}}$ by merging two occurrences $\gamma
\rightarrow \beta $ using the repetition $\left( R\right) _{2}$ with
conclusion $\beta $ : \newline

{\small \hspace{-20pt} 
\begin{tikzpicture}
  \draw (0,2.5) node {
    \AxiomC{$[\alpha]^{1}$}
    \AxiomC{$[\alpha \rightarrow\beta]^{2}$}
    \LeftLabel{$\left\{ \textcolor{red}{1}\right\}$}
    \RightLabel{$\left\{ \textcolor{red}{2}\right\}$}
    \BinaryInfC{$\beta $}
    \AxiomC{$[\delta]^{3}$}
    \AxiomC{$[\delta\rightarrow\beta]^{4}$}
    \LeftLabel{$\left\{ \textcolor{red}{3}\right\}$}
    \RightLabel{$\left\{ \textcolor{red}{4}\right\}$}
    \BinaryInfC{$\beta$} 
    \LeftLabel{$\left\{ \textcolor{red}{1},\textcolor{red}{2}\right\}$}
    \RightLabel{$\left\{\textcolor{red}{3}, \textcolor{red}{4}\right\}$}
    \BinaryInfC{$\beta$}
    \LeftLabel{$\left\{\textcolor{red}{1}, \textcolor{red}{2}, \textcolor{red}{3}, \textcolor{red}{4}\right\}$}
    \UnaryInfC{$\gamma \rightarrow \beta$} 
     \LeftLabel{$\left\{ \textcolor{red}{1},\textcolor{red}{2}, \textcolor{red}{3}, \textcolor{red}{4}\right\}$}
    \UnaryInfC{$$} 
       \DisplayProof
  } ;
  \draw [->] (0,1.5) -- (2.5,1) ;
  \draw [->] (0,1.5) -- (-2.5,1) ;
  \draw (1,0) node {
      \AxiomC{}
    \LeftLabel{$\left\{  \textcolor{red}{1},\textcolor{red}{2}, 3, 4\right\}$\!}\!
    \UnaryInfC{\!$ (\alpha\!\rightarrow \!\beta) \!\rightarrow\!(\gamma \!\rightarrow \!\beta)$}
    \LeftLabel{$\left\{ \textcolor{red}{1}, \! 2, \! 3, \! 4\right\}$\!}\!
    \UnaryInfC{\!$\alpha\!\rightarrow \!((\alpha\!\rightarrow \!\beta)\! \rightarrow\!(\gamma\! \rightarrow \!\beta))$}
    \LeftLabel{$\left\{ 1, \! 2, \! 3, \! 4\right\}$\!}\!
    \UnaryInfC{$\alpha\!\rightarrow \!((\alpha\rightarrow\! \beta)\! \rightarrow\!(\gamma\! \rightarrow\! \beta))$}
    \AxiomC{}
    \RightLabel{$\left\{1, \! 2, \! \textcolor{red}{3}, \! \textcolor{red}{4}\right\}$}    
    \UnaryInfC{\!\!$(\delta\!\rightarrow\! \beta) \!\rightarrow\!(\gamma \!\rightarrow \!\beta)$}
    \LeftLabel{$\!\!\left\{ 1,  \!2,  \!\textcolor{red}{3}, \!4\right\}$\!}\! 
   \UnaryInfC{$\!\!\delta\! \rightarrow\! ((\delta\!\rightarrow\! \beta)\! \rightarrow\!(\gamma \!\rightarrow \!\beta))$}
    \AxiomC{$\QDATOP{\left[ \tau \right]^{5} }{ \,\downarrow }$}
    \RightLabel{$\left\{\textcolor{red}{5}\right\}$}
    \UnaryInfC{$\left. \overset{}\tau \right.$}
    \RightLabel{$\left\{ \textcolor{red}{5}\right\}$}
    \UnaryInfC{$\left. \overset{%
}\tau \right.$}
    \LeftLabel{$\left\{ 1, \! 2, \! 3, \! 4\right\}$\!}\!
    \RightLabel{$\left\{ \textcolor{red}{5}\right\}$}
    \BinaryInfC{$(\alpha\!\rightarrow\! ((\alpha\! \rightarrow\! \beta)\! \rightarrow\! (\gamma\!\rightarrow\!\beta))) \!\rightarrow \!\sigma$}
    \LeftLabel{$\left\{ 1, \! 2, \! 3, \! 4\right\}$\!}\!
    \RightLabel{$\left\{1,  \!2,  \!3,  \!4,  \!\textcolor{red}{5}\right\}$}
    \BinaryInfC{$\sigma$}
    \LeftLabel{$\left\{1,  \!2,  \!3,  \!4,  \!\textcolor{red}{5}\right\}$} 
    \UnaryInfC{$\tau \!\rightarrow \!\sigma$} 
    \LeftLabel{$\left\{1,  \!2,  \!3,  \!4,  \!5\right\}$} 
   \UnaryInfC{$$}
   
 \DisplayProof
  } ;
\end{tikzpicture}
\newline
} for $A=\left\{ \alpha ,\alpha \!\rightarrow \!\beta ,\delta ,\delta
\!\rightarrow \!\beta ,\tau \right\} $ (encoded as $\left\{
1,\!2,\!3,\!4,\!5\right\} $), while $f\!:E\times E\times A\!\rightarrow
\!\wp \left( E\right) $ is such that for any $e_{1}=\left\langle
x_{1},y_{1}\right\rangle \prec e_{2}=\left\langle x_{2},y_{2}\right\rangle $
exposed and any $\xi \in A\left( e_{2}\right) $, $f\!\left( e_{1},e_{2},\xi
\right) =IE\left( e_{2},\xi \right) $ if either $\left\lceil
x_{2}\right\rceil \!\neq \!\left( R\right) _{2}$ or else $e_{1}$ is one of
the inferences with conclusions $\tau $ or $\tau \rightarrow \sigma $.
Otherwise, either $f\!\left( e_{1},e_{2},\xi \right) =e^{\prime }$ or $%
f\!\left( e_{1},e_{2},\xi \right) =e^{\prime \prime }$, provided that $%
\left\lceil x_{2}\right\rceil \!=\left( R\right) _{2}$ with $e^{\prime }$
and $e^{\prime \prime }$ being respectively the left- and right-hand side
edge in $IE\left( e_{2},\xi \right) $ and $e_{1}$ is placed in $\partial
_{f}^{\prime }$ on the same left- or right-hand side under $e_{2}$. The
assumptions $A\left( e_{i}\right) $ are exposed in curly braces at the
corresponding inferences, while $A_{f}\left( e_{i}\right) $ colored red.

It is readily seen that $A_{f}(r)=\emptyset $ in $\partial _{f}^{\prime }$,
and hence $\partial _{f}^{\prime }\vdash \tau \rightarrow \sigma $, while $%
\partial $ is the tree-like unfolding of $\partial _{f}^{\prime }$. By
contrast, $\partial ^{\prime }$ (without $f$) has open paths connecting
pairs of edges supplied with the assumptions $\left\{ \text{%
\textcolor{red}{1},\textcolor{red}{2}}\right\} ,\left\{ 1,\!2,\!\text{%
\textcolor{red}{3}},\text{\textcolor{red}{4}}\right\} $ and/or $\left\{ 
\text{\textcolor{red}{3},\textcolor{red}{4}}\right\} ,\left\{ \text{%
\textcolor{red}{1},\textcolor{red}{2}},3,4\right\} $.
\end{example}

\begin{lemma}
Problem $\partial _{f}\vdash ^{?}\!\rho $ is verifiable by deterministic TM
in time polynomial in the weight of $\partial $, $\underset{x\in V}{\sum }%
\left| \ell \left( x\right) \right| $. \footnote{{\small {\footnotesize Note
that the weight of $\partial $ exceeds its size $\left| \partial \right| $%
\thinspace = the total number of nodes in $\partial $.}}}
\end{lemma}

\begin{proof}
The required $f$-regularity (cf. Definition 4) is easily verifiable in $%
\left| \partial \right| $-polynomial time using polynomial complexity of the
reachability relation $\prec $, cf. e.g. \cite{Papa}. As for $\partial
_{f}\vdash \rho $ proper, we use Lemma 6 and verify in polynomial time the
condition $A_{f}(r)=\emptyset $. To this end, it will suffice to construct a
rooted Boolean circuit $C_{f}$ of $\left| V\right| $-polynomial size with
inputs from $A$ such that for any $\alpha \in A$, $C_{f}\left( \alpha
\right) =0\Leftrightarrow \alpha \in A_{f}(r)$. For the sake of brevity
first consider the simplest ``conventional'' case where $f\!\left(
e_{1},e_{2},\alpha \right) =IE\left( e_{2},\alpha \right) $ holds for all $%
e_{1}$, $e_{2}$, $\alpha $ involved. \footnote{{\small {\footnotesize In
this case, all maximal deductive paths are closed.}}} Let $C_{0}$ be the
circuit-skeleton of $\partial $ whose leaves are assigned with the
underlying assumptions and consider any chosen $\alpha \in A$. Then let $C$
be a Boolean extension of $C_{0}$ that is obtained by assigning nodes $x\in
C_{0}$ with Boolean values $\varepsilon \left( x\right) \in \left\{
0,1\right\} $ according to the following recursive clauses 1--3 and let $%
C\left( \alpha \right) :=\varepsilon \left( r\right) $.

\begin{enumerate}
\item  If $h\left( x\right) =h$ then $\varepsilon \left( x\right) :=\left\{ 
\begin{array}{lll}
1, & \text{if} & \alpha =\ell \left( x\right) , \\ 
0, & \text{else.} & 
\end{array}
\right. $

\item  If $h\left( x\right) <h$ and $\left\lceil x\right\rceil \neq \left(
\rightarrow I\right) $, then $\varepsilon \left( x\right) :=\bigvee \left\{
\varepsilon \left( y\right) :\left\langle y,x\right\rangle \in E\right\} $.

\item  If $h\left( x\right) <h$ and $\left\lceil x\right\rceil =\left(
\rightarrow _{\beta }I\right) $ with $\left\langle y,x\right\rangle \in E$,
then \newline
$\varepsilon \left( x\right) :=\left\{ 
\begin{array}{lll}
0, & \text{if} & \alpha =\beta , \\ 
\varepsilon \left( y\right) , & \text{else.} & 
\end{array}
\right. $
\end{enumerate}

Now consider general case of arbitrary flow function $f\!:E\times E\times
A\!\rightarrow \!\wp \left( E\right) $. A desired Boolean extension of $%
C_{f} $ is defined analogously, except replacing clauses $2$ and $3$ by the
following $2^{\prime }$ and $3^{\prime }.$

\begin{description}
\item  $2^{\prime }$. If $h\left( x\right) <h$ and $\left\lceil
x\right\rceil \neq \left( \rightarrow _{\alpha }I\right) $, then

$\varepsilon \left( x\right) :=\bigvee \left\{ \varepsilon \left( y^{\prime
\prime }\right) :\left( \exists e\in IE\left( x\right) \right) \left(
\exists e^{\prime }\succ e\right) \left( \exists e^{\prime \prime
}=\left\langle y^{\prime \prime },x^{\prime \prime }\right\rangle \in
f\left( e,e^{\prime },\alpha \right) \right) \right\} .$

\item  $3^{\prime }$. If $h\left( x\right) <h$ and $\left\lceil
x\right\rceil =\left( \rightarrow _{\alpha }I\right) $, then $\varepsilon
\left( x\right) :=0$.
\end{description}

The resulting query $C_{f}\left( \alpha \right) \overset{?}{=}0$ is
decidable in $\left| \partial \right| $-polynomial time, as $h,\left|
A\right| ,\left| IE\left( x\right) \right| \leq \!\left| V\right| $, $\left|
E\right| \leq \!\left| V\right| ^{2}$ and $\left| f\right| \leq \!\left|
V\right| ^{6}$.
\end{proof}

\begin{definition}
A given tree-like or dag-like deduction with conclusion $\rho $ is called 
\emph{polynomial}, resp. \emph{quasi-polynomial}, if its weight, resp.
height plus total weight of distinct formulas involved, is polynomial in the
size of $\rho $, $\left| \rho \right| $.
\end{definition}

Our crucial result reads

\begin{theorem}[Main Theorem]
Any quasi-polynomial tree-like proof $\partial ^{\text{\textsc{t}}}$ of $%
\rho $ in \textsc{NM}$_{\rightarrow }$ is compressible into a polynomial
regular dag-like proof $\partial _{f}^{\flat }$ of $\rho $ in \textsc{NM}$%
_{\rightarrow }^{\text{\textsc{d}}}$.
\end{theorem}

The required mapping $\partial ^{\text{\textsc{t}}}\hookrightarrow \partial
_{f}^{\flat }$ is obtained by the horizontal compression\ as explained below.

\subsection{Horizontal compression}

In the sequel for any natural deduction $\partial $ we denote by $h\left(
\partial \right) $ and $\phi \left( \partial \right) $ the height of $%
\partial $ and the total weight of the set of distinct formulas occurring in 
$\partial $, respectively. Now we are prepared to explain proof of the Main
Theorem. For any quasi-polynomial tree-like proof $\partial ^{\text{\textsc{t%
}}}$ of $\rho \,$ let $\partial ^{\prime }$ be its horizontal compression
defined by bottom-up recursion on $h\left( \partial ^{\text{\textsc{t}}%
}\right) $ such that for any $k\leq h\left( \partial ^{\text{\textsc{t}}%
}\right) $, the $k^{th}$\ horizontal section of $\partial ^{\prime }$ is
obtained by merging all nodes with identical formulas occurring in the $%
k^{th}$\ horizontal section of $\partial ^{\text{\textsc{t}}}$. The
inferences in $\partial ^{\prime }$ are naturally inherited by the ones in $%
\partial ^{\text{\textsc{t}}}$.\thinspace \footnote{{\small {\footnotesize %
See \cite{GH1}.}}} Obviously $\partial ^{\prime }$ is a dag-like (not
necessarily tree-like anymore) deduction with conclusion $\rho $. Moreover $%
\partial ^{\prime }$ is polynomial as $\left| \partial ^{\prime }\right|
\leq h\left( \partial ^{\text{\textsc{t}}}\right) \times \phi \left(
\partial ^{\text{\textsc{t}}}\right) $. However, $\partial ^{\prime }$ need
not preserve basic inferences $\left( \rightarrow I\right) $, $\left(
\rightarrow E\right) $. For example, a compressed multipremise configuration 
\begin{equation*}
\fbox{$\ \left( \rightarrow I,E\right) ^{\prime }:\dfrac{\beta \quad \quad
\gamma \quad \quad \gamma \rightarrow \left( \alpha \rightarrow \beta
\right) }{\alpha \rightarrow \beta }$}
\end{equation*}
that is obtained by merging identical conclusions $\alpha \rightarrow \beta $
of 
\begin{equation*}
\fbox{$\ \left( \rightarrow I\right) :\dfrac{\beta }{\alpha \rightarrow
\beta }$}\quad \text{and\quad }\fbox{$\ \left( \rightarrow E\right) :\dfrac{%
\gamma \quad \quad \gamma \rightarrow \left( \alpha \rightarrow \beta
\right) }{\alpha \rightarrow \beta }$}
\end{equation*}
is not a correct inference in \textsc{NM}$_{\rightarrow }$ and it should be
replaced by a \textsc{NM}$_{\rightarrow }^{\text{\textsc{d}}}$ inference 
\begin{equation*}
\fbox{$\!\left( S\right) \ \dfrac{\dfrac{\beta }{\alpha \rightarrow \beta }%
\quad \quad \dfrac{\gamma \quad \quad \gamma \rightarrow \!\left( \alpha
\!\rightarrow \!\beta \right) }{\alpha \rightarrow \beta }}{\alpha
\rightarrow \beta }$}.
\end{equation*}
So in order to obtain a correct dag-like deduction we replace all
multipremise configurations by the corresponding instances of $\left(
S\right) $, which converts $\partial ^{\prime }$ into a locally correct
dag-like deduction $\partial ^{\text{\textsc{d}}}\in \,$\textsc{NM}$%
_{\rightarrow }^{\text{\textsc{d}}}$ with the same conclusion $\rho $. Let $%
h:=h\left( \partial ^{\text{\textsc{d}}}\right) $. The weight of $\partial ^{%
\text{\textsc{d}}}$ is still polynomial, as $h\leq 2h\left( \partial ^{\text{%
\textsc{t}}}\right) \ $and $\left| \partial ^{\text{\textsc{d}}}\right| \leq
h\times \phi ^{2}\left( \partial ^{\text{\textsc{t}}}\right) $. Note that
arbitrary deductive dag-like paths in $\partial ^{\text{\textsc{d}}}$ can
arise by concatenating different segments of maximal tree-like deductive
paths in $\partial $ (which are closed by the assumption). Let $\mathcal{F}^{%
\text{\textsc{d}}}$ be the dag-like image of the set of these maximal
tree-like deductive paths under the mapping $\partial \hookrightarrow
\partial ^{\text{\textsc{d}}}$. We observe that $\mathcal{F}^{\text{\textsc{d%
}}}$ satisfies the following three conditions of \emph{local coherency},
where for any maximal dag-like deductive path $\Theta =\left[ x_{0},\cdots
,x_{h}=r\right] \in \mathcal{F}^{\text{\textsc{d}}}$ in $\partial ^{\text{%
\textsc{d}}}$ and any $i<h$ we let $\Theta \!\downharpoonright _{i}:=\left[
x_{i},\cdots ,x_{h}\right] $. As above,\ we call $\Theta $ \emph{closed} if $%
\left( \exists i<h\right) \!\left\lceil x_{i}\right\rceil \!=\!\left(
\rightarrow _{\ell \left( x_{0}\right) }I\right) $.

\begin{enumerate}
\item  $\mathcal{F}^{\text{\textsc{d}}}$ is \emph{dense} in $\partial ^{%
\text{\textsc{d}}}$, i.e. $\left( \forall x\in V\left( \partial ^{\text{%
\textsc{d}}}\right) \right) \left( \exists \Theta \in \mathcal{F}^{\text{%
\textsc{d}}}\right) \left( x\in \Theta \right) $.

\item  $\mathcal{F}^{\text{\textsc{d}}}$ is \emph{closed}, i.e. so is every $%
\Theta \in \mathcal{F}^{\text{\textsc{d}}}$.

\item  $\mathcal{F}^{\text{\textsc{d}}}$ \emph{preserves} $\left(
\rightarrow E\right) $, i.e. for any $i<h$ and $x_{i},x_{i+1}\in \Theta \in 
\mathcal{F}^{\text{\textsc{d}}}$ such that $\left\lceil x_{i+1}\right\rceil
=\left( \rightarrow E\right) $ with $I\!E\left( x_{i+1}\right) \!=\!\left\{
\left\langle x_{i},x_{i+1}\right\rangle ,\left\langle y,x_{i+1}\right\rangle
\right\} $ (i.e. $x_{i}\!\neq \!y$ are the \emph{children }of $x_{i+1}$),
there is a $\Theta ^{\prime }\in \mathcal{F}^{\text{\textsc{d}}}$ such that $%
y\in \Theta ^{\prime }$ and $\Theta \!\downharpoonright _{i+1}=\Theta
^{\prime }\!\downharpoonright _{i+1}.$
\end{enumerate}

Now for any given subdeduction $\partial ^{\flat }\subseteq \partial ^{\text{%
\textsc{d}}}$, any subset $\mathcal{F}^{\flat }\subseteq \mathcal{F}^{\text{%
\textsc{d}}}$ that is dense (in $\partial ^{\text{\textsc{d}}}$), closed and
preserves $\left( \rightarrow E\right) $ is called a \emph{fundamental set
of paths} (abbr.: \emph{fsp}) in $\partial ^{\flat }$. To complete our proof
of the Main Theorem it will suffice to prove

\begin{lemma}
There exist a $\partial ^{\flat }\subseteq \partial ^{\text{\textsc{d}}}$
with a flow $\!f:E^{\flat }\times E^{\flat }\times A^{\flat }\!\rightarrow
\!\wp \left( E^{\flat }\right) $ and a fsp $\mathcal{F}^{\flat }\subseteq 
\mathcal{F}^{\text{\textsc{d}}}$ in $\partial ^{\flat }$ such that $\partial
_{f}^{\flat }$ is regular and has no open $f$-threads.
\end{lemma}

\begin{proof}
The required $\mathcal{F}^{\flat }$ is defined by simultaneous ascending
recursion. To put it more exactly, we let $\mathcal{F}^{\flat }:=\overset{h}{%
\underset{k=1}{\bigcup }}\mathcal{F}_{k}^{\flat }$ where $\mathcal{F}%
_{k}^{\flat }\subseteq \mathcal{F}^{\text{\textsc{t}}}$ arise by recursion
on $0<k<h$ such that $\mathcal{F}_{k}^{\flat }\subseteq \mathcal{F}%
_{k+1}^{\flat }$ and for any $\Theta ^{\prime }\in \mathcal{F}_{k+1}^{\flat
} $ there exists $\Theta \in \mathcal{F}_{k}^{\flat }$\ with $\Theta
\!\downharpoonright _{k+1}=\Theta ^{\prime }\!\downharpoonright _{k+1}$.
Having done this we stipulate $\partial ^{\flat }$ by collecting all edges
occurring in $\mathcal{F}^{\flat }$ and denote by $A^{\flat }$ the set of
assumptions occurring in $\mathcal{F}^{\flat }$. Turning back to $\mathcal{F}%
_{n}^{\flat }$ the recursion runs as follows.

\emph{Basis of recursion}. Suppose $k=1$. We have $\left\lceil r\right\rceil
\in \left\{ \left( \rightarrow I\right) ,\left( \rightarrow E\right)
\right\} $. If $\left\lceil r\right\rceil =\left( \rightarrow I\right) $
with $I\!E\left( r\right) =\left\{ \left\langle x,r\right\rangle \right\} $\
then (by the density of $\mathcal{F}^{\text{\textsc{d}} })$ choose any $%
\Theta \in \mathcal{F}^{\text{\textsc{d}}}$ with $x\in \Theta$ and let $%
\mathcal{F}_{1}^{\flat }:=\left\{ \Theta \right\} $. Otherwise, if $%
\left\lceil r\right\rceil =\left( \rightarrow E\right) $ with $I\!E\left(
r\right) =\left\{ \left\langle x,r\right\rangle ,\left\langle
y,r\right\rangle \right\} $ then by the same token let $\mathcal{F}%
_{1}^{\flat }:=\left\{ \Theta _{1},\Theta _{2}\right\} $, for any chosen $%
\Theta _{1},\Theta _{2}\in \mathcal{F}^{\ast }$ with $x\in \Theta _{1},y\in
\Theta _{2}$.

\emph{Recursion step. }Suppose $0<k<h$ and let $i:=h-k$. Then $\mathcal{F}%
_{k+1}^{\flat }$ extends $\mathcal{F}_{k}^{\flat }$ in the following way.
Consider any $\Theta =\left[ x_{0},\cdots ,x_{h}=r\right] \in \mathcal{F}%
_{k}^{\flat }$\ and suppose that $\left\lceil x_{i+1}\right\rceil \neq
\left( \rightarrow E\right) $. Then let $\Theta \in \mathcal{F}_{k+1}^{\flat
}$. Otherwise, if $\left\lceil x_{i+1}\right\rceil =\left( \rightarrow
E\right) $ with $I\!E\left( x_{i+1}\right) =\left\{ \left\langle
x_{i},x_{i+1}\right\rangle ,\left\langle y,x_{i+1}\right\rangle \right\} $
then (since $\mathcal{F}^{\text{\textsc{d}}}$ preserves $\left( \rightarrow
E\right) $) choose any $\Theta ^{\prime }\in \mathcal{F}^{\text{\textsc{d}}}$
such that $y\in \Theta ^{\prime }$ and $\Theta \!\downharpoonright
_{i+1}=\Theta ^{\prime }\!\downharpoonright _{i+1}$ and hen let $\Theta
^{\prime }\in \mathcal{F}_{k+1}^{\flat }$.

This completes the recursive definition of the \emph{fsp} $\mathcal{F}%
^{\flat }$ and corresponding $\partial ^{\flat }$ (see above). To get the
required flow $\!f:E^{\flat }\times E^{\flat}\times A^{\flat }\!\rightarrow
\!\wp \left( E^{\flat }\right) $ consider any $e_{1}\preceq e_{2}$ and $%
\alpha \in A\left( e_{2}\right) $ and set $f(e_{1},e_{2},\alpha ):=\smallskip \\
\left\{ e\in E^{\flat }:\left. 
\begin{array}{c}
\left( \exists \Theta =\left[ x_{0},\cdots ,x_{h}=r\right] \in \mathcal{F}%
^{\flat }\right) \left( \exists 0<i\leq j<h\right) \\ 
e_{1}=\left\langle x_{j},x_{j+1}\right\rangle \,\&\, \, e_{2}=\left\langle
x_{i},x_{i+1}\right\rangle \,\&\, \, e=\left\langle
x_{i-1},x_{i}\right\rangle \,\&\,\,\alpha =\ell \left( x_{0}\right)
\end{array}
\right. \right\} .$

Regular correctness of $\partial _{f}^{\flat }$ is verified by
straightforward induction on the depth of $x$. To complete the whole proof
it remains to observe that by the definition of $\mathcal{F}^{\flat }$,
every $f$-thread $\Theta $ in $\partial _{f}^{\flat }$ belongs to $\mathcal{F%
}^{\text{\textsc{d}}}$, thus being closed, as required.
\end{proof}

This lemma completes our proof of Theorem 11.

\section{Consequences for computational complexity}

Consider the well-known complexity classes of computational problems \textbf{%
NP}, \textbf{coNP} and \textbf{PSPACE} (see \S 1.1). Recall that 
\textbf{NP} and \textbf{coNP} are both contained in \textbf{PSPACE},
although precise comparison between these classes is often referred to as
fundamental open problem (cf. e.g. \cite{arora}, \cite{Papa}). It is known,
however, that \textbf{PSPACE} can be characterized by the provability in
minimal logic ( \cite{Statman}, \cite{Svejdar}) and HAMILTON PATH problem is 
\textbf{NP} complete, while $\mathbf{NP=coNP}$ is a consequence of $\mathbf{%
NP=PSPACE}$ (see e.g. \cite{arora}, \cite{Papa}). Our proof of the
equalities $\mathbf{NP=coNP=PSPACE}$ use these results. In fact, by the last
one it will suffice to prove the main equality $\mathbf{NP=PSPACE}$ (however
see \S 3.2 below).

\subsection{ General case $\mathbf{NP}$ vs $\mathbf{PSPACE}$}

Recall that we consider standard language of minimal logic whose formulas ($%
\alpha $, $\beta $, $\gamma $, $\rho $ etc.) are built up from propositional
variables ($p$, $q$, $r$, etc.) using one propositional connective `$%
\rightarrow $'. The sequents are in the form $\Gamma \Rightarrow \alpha $\
whose antecedents, $\Gamma $,\ are viewed as multisets of formulas; sequents 
$\Rightarrow \alpha $\ , i.e. $\emptyset \Rightarrow \alpha $, are
identified with formulas $\alpha $.

Our special sequent calculus for minimal logic, \textsc{LM}$_{\rightarrow }$%
, includes the following axioms $\left( \text{\textsc{M}}A\right) $ and
inference rules $\left( \text{\textsc{M}}I1\rightarrow \right) $, $\left( 
\text{\textsc{M}}I2\rightarrow \right) $, $\left( \text{\textsc{M}}%
E\rightarrow P\right) $, $\left( \text{\textsc{M}}E\rightarrow \rightarrow
\right) $ in the language $\mathcal{L}_{\rightarrow }$ (the constraints are
shown in square brackets). \footnote{{\small {\footnotesize This is a
slightly modified, equivalent version of the corresponding \ purely
implicational and }$\bot ${\footnotesize -free} {\footnotesize subsystem of
Hudelmaier's intuitionistic calculus \textsc{LG}, cf. \cite{Hud}. The
constraints }$q\in VAR\left( \Gamma ,\gamma \right) ${\footnotesize \ are
added just for the sake of transparency.}}}\medskip

$\fbox{$\left( \text{\textsc{M}}A\right) :\ \ \ \Gamma ,p\Rightarrow p$}$

$\fbox{$\left( \text{\textsc{M}}I1\!\rightarrow \right) :\ \ \ \dfrac{\Gamma
,\alpha \Rightarrow \beta }{\Gamma \Rightarrow \alpha \rightarrow \beta }%
\smallskip \quad \left[ \left( \nexists \gamma \right) :\left( \alpha
\rightarrow \beta \right) \rightarrow \gamma \in \Gamma \right] $}$

$\fbox{$\left( \text{\textsc{M}}I2\!\rightarrow \right) :\ \ \ \dfrac{\Gamma
,\alpha ,\beta \rightarrow \gamma \Rightarrow \beta }{\Gamma ,\left( \alpha
\rightarrow \beta \right) \rightarrow \gamma \Rightarrow \alpha \rightarrow
\beta }$}$

$\fbox{$\left( \text{\textsc{M}}E\!\rightarrow \!P\right) :\ \ \ \dfrac{%
\Gamma ,p,\gamma \Rightarrow q}{\Gamma ,p,p\rightarrow \gamma \Rightarrow q}%
\quad \left[ q\in \mathrm{VAR}\left( \Gamma ,\gamma \right) ,p\neq q\right] $%
}$

$\fbox{$\left( \text{\textsc{M}}E\!\rightarrow \rightarrow \right) :\ \ \ 
\dfrac{\Gamma ,\alpha ,\beta \rightarrow \gamma \Rightarrow \beta \quad
\quad \Gamma ,\gamma \Rightarrow q}{\Gamma ,\left( \alpha \rightarrow \beta
\right) \rightarrow \gamma \Rightarrow q}\quad \left[ q\in \mathrm{VAR}%
\left( \Gamma ,\gamma \right) \right] $}$

\begin{claim}
\textsc{LM}$_{\rightarrow }$ is sound and complete with respect to minimal
propositional logic and tree-like deducibility and any purely implicational
formula $\rho $ is valid in minimal logic iff corresponding sequent $%
\Rightarrow \rho $ is provable by a quasi-polynomial tree-like derivation
(deduction) in \textsc{LM}$_{\rightarrow }$.
\end{claim}

\begin{proof}
This easily follows from \cite{Hud} (see \cite{GH1} for details).
\end{proof}

\begin{lemma}
For any \textsc{LM}$_{\rightarrow }$-valid purely implicational formula $%
\rho $ there exists a quasi-polynomial tree-like proof of $\rho$ in \textsc{%
NM}$_{\rightarrow }$.
\end{lemma}

\begin{proof}
This follows by a straightforward interpretation in \textsc{NM}$%
_{\rightarrow }$ of a proof in \textsc{LM}$_{\rightarrow }$ that must exist
by the validity of $\rho $ (see \cite{GH1}, \cite{GH2} for details).
\end{proof}

\begin{theorem}
$\mathbf{PSPACE}\subseteq \mathbf{NP}$ and hence $\mathbf{NP=PSPACE}$ holds
true.
\end{theorem}

\begin{proof}
The inclusion $\mathbf{NP\subseteq PSPACE}$ is well known (cf. e.g. \cite
{arora}, \cite{Papa}). Hence it suffices to prove $\mathbf{PSPACE}\subseteq 
\mathbf{NP}$. To this end we recall that minimal propositional logic is 
\textbf{PSPACE}-complete (cf. e.g. \cite{Joh}, \cite{Statman}, \cite{Svejdar}%
). Thus by Claim 13, in order to\ arrive at $\mathbf{PSPACE}\subseteq 
\mathbf{NP}$ it will suffice to show that the validity problem VAL in
minimal logic is in $\mathbf{NP}$. In particular, by Lemma 6, it will
suffice to show that for any given purely implicational formula $\rho $ that
is valid in \textsc{LM}$_{\rightarrow }$ there exists a polynomial regular
dag-like proof of $\rho $\ in \textsc{NM}$_{\rightarrow }$. But this follows
from Theorem 8 together with Lemma 14.

Summing up, we have shown that for any proof in \textsc{LM}$_{\rightarrow }$
of a given formula $\rho $ there exist, successively: (1) a quasi-polynomial
tree-like proof $\partial ^{\text{\textsc{t}}}\in \,$\textsc{NM}$%
_{\rightarrow }^{\text{\textsc{t}}}$ of $\rho $, (2) a polynomial dag-like
deduction $\partial ^{\text{\textsc{d}}}\in \,$\textsc{NM}$_{\rightarrow }^{%
\text{\textsc{d}}}$ of $\rho $, (3) a subdeduction (hence polynomial) $%
\partial ^{\flat }\subseteq \partial ^{\text{\textsc{d}}}$ with a suitable
flow $f$ that eventually yields a desired polynomial regular dag-like proof $%
\partial _{f}^{\ast }:=\partial _{f}^{\flat }$ of $\rho $. Since provability
(= validity) in \textsc{LM}$_{\rightarrow }$ is $\mathbf{PSPACE}$-complete
and $\partial _{f}^{\flat }$ has polynomial certificate (cf. Lemma 9), this
completes our proof of $\mathbf{PSPACE}\subseteq \mathbf{NP}$.
\end{proof}

\begin{corollary}
The satisfiability and validity problems in quantified boolean logic (QBL)
are both $\mathbf{NP}$-complete. Moreover $\mathbf{BQP\subseteq NP}$ holds
true, where $\mathbf{BQP}$ is the class of problems computable in quantum
polynomial time.
\end{corollary}

\begin{proof}
The former assertion follows from the $\mathbf{PSPACE}$-completeness of QBL
(see e.g. \cite{arora}, \cite{Papa}). The latter follows from $\mathbf{%
BQP\subseteq PSPACE}$ (cf. \cite{arora}).
\end{proof}

\subsection{Practical implementation}

Horizontal compression of a given tree-like ND proof $\partial $ can be
executed by a deterministic TM in time polynomial in the weight of $\partial 
$; in most interesting cases the compressed dag-like deduction $\partial ^{%
\text{\textsc{d}}}$ will be exponentially smaller than $\partial$ (cf. \cite
{Haeus1}). This procedure is being implemented in LEAN (\cite{Rob}, work in
progress). However, the subsequent operation $\partial ^{\text{\textsc{d}}%
}\hookrightarrow \partial_{f}^{\flat }$ (cf. Lemma 12) is less constructive
due to nondeterministic condition 3 of the local coherency of $\mathcal{F}^{%
\text{\textsc{d}}}$ involved. To put it bluntly, our proof of crucial
inclusion $\mathbf{PSPACE}\subseteq \mathbf{NP}$ appears less deterministic
than (say) standard proof of VAL\,$\in \mathbf{NP}$.

\subsection{Particular case $\mathbf{NP}$ vs $\mathbf{coNP}$}

Our proof of the main result $\mathbf{NP\!=\!PSPACE}$ (above) referred to
special sequent calculus for minimal logic, \textsc{LM}$_{\rightarrow }$. In
the case of weaker equality $\mathbf{NP=coNP}$ this reference is obsolete.
Instead, we'll refer to tree-like ND-normalization and recall that \emph{%
normal} tree-like ND proofs satisfy weak subformula property (see \S 1.2).
Moreover, the following holds (see \cite{GH3}, \cite{Haeus}).

\begin{lemma}
Any normal\ tree-like \textsc{NM}$_{\rightarrow }$-proof of $\rho $ whose
height is polynomial in $\left|\rho \right| $ is quasi-polynomial.
\end{lemma}

Now let purely implicational formula $\rho $ express in standard way that a
given graph $G$ has no Hamiltonian cycles. Thus $\rho $ is valid iff
corresponding NP-complete problem HAMILTON PATH has no positive solution. We
also observe that the canonical proof search of $\rho $ in \textsc{NM}$%
_{\rightarrow }$ yields a normal tree-like proof $\partial $ whose height is
polynomial in $\left| G\right| $ (and hence in $\left| \rho \right| $),
provided that $G$ is non-Hamiltonian. Further we argue via

\begin{lemma}[cf. \protect\cite{GH3}]
Consider the HAMILTON PATH problem and a purely implicational formula $\rho$
expressing that a given graph $G$ has no Hamiltonian cycles. There exists a
normal tree-like \textsc{NM}$_{\rightarrow }$-proof $\partial$ of $\rho$
such that $h\left( \partial \right) $ is polynomial in $\left| G\right| $
(and in $\left| \rho\right| $), provided that $G$ is non-Hamiltonian.
\end{lemma}

Recall that polynomial ND proofs (whether tree-like or regular dag-like)
have polynomial-time certificates, while the non-hamiltoniancy of simple and
directed graphs is $\mathbf{coNP}$-complete. Hence Theorem 11 yields

\begin{corollary}
$\mathbf{NP=coNP}$\textbf{\ }holds true.
\end{corollary}

\section{Conclusion}

Our main result, $\mathbf{NP=PSPACE}$, is obtained by a proof theoretic
approach whose crucial part deals with proof compression in Prawitz-style
formalism of Natural Deduction (ND). Its starting point is based on a desire
to transform the validity problem $\rho \in VAL$, for a given purely
implicational formula $\rho $, into ``small'' and verifiable in $O(\left|
\rho \right| )$-polynomial time regular proof $\partial _{f}^{\flat }$ in
our dag-like ND calculus \textsc{NM}$_{\rightarrow }^{\text{\textsc{d}}}$
for minimal propositional logic (see \S 2.2). To this end we take an
ordinary tree-like derivation $\partial $ of sequent $\Rightarrow \rho $ in
Hudelmaier's sound and complete sequent calculus \textsc{LM}$_{\rightarrow }$
(see \S 3.2) and convert $\partial $ into quasi-polynomial deduction $%
\partial^{\text{\textsc{t}}}$ of $\rho $ in the corresponding tree-like ND
calculus \textsc{NM}$_{\rightarrow }^{\text{\textsc{t}}}$. To complete the
construction by Theorem 11 we horizontally compress $\partial ^{\text{%
\textsc{t}}}$ into polynomial dag-like deduction $\partial ^{\text{\textsc{d}%
}}$ of $\rho $ in \textsc{NM}$_{\rightarrow }^{\text{\textsc{d}}}$ followed
by a desired regular proof $\partial _{f}^{\flat }\subseteq \partial ^{\text{%
\textsc{d}}}$ (recall that Theorem 11 says that all minimal tautologies $%
\rho $ are provable by regular dag-like proofs in \textsc{NM}$_{\rightarrow
}^{\text{\textsc{d}}}$ whose weight is polynomial in the size of $\rho $).
In particular this shows that our multipremise dag-like calculus \textsc{NM}$%
_{\rightarrow }^{\text{\textsc{d}}}$ is sound and complete in minimal logic
with respect to regular dag-like provability.

In the presence of multipremise repetition involved (see \S \S 2.1, 2.2),
regular dag-like provability is crucial for Theorem 11 and resulting
polynomial upper bounds on the size (weight) of proofs in minimal logic.
This is emphasized by Je\v{r}\'{a}bek's paper \cite{Jer} that claims to
obtain exponential lower bounds on the size of dag-like proofs with respect
to conventional notion of provability determined by the condition \emph{%
``all deductive paths are closed''}, instead of our weakening \emph{``all
regular deductive paths are closed''}. According to \cite{Jer}, these
exponential lower bounds hold true also for the Frege systems involved,
which therefore appear useless for the comparison $\mathbf{NP}$ vs $\mathbf{%
coNP}$ vs $\mathbf{PSPACE}$ under consideration. Summing up, exponential
lower bounds claimed in \cite{Jer} are well compatible with our polynomial
upper bounds in the case of dag-like multipremise ND. To put it more
precisely, there are polynomial-time dag-like proofs in our ND formalism
which have no analogy in Frege systems from \cite{Jer}. This provides
exponential speed-up over conventional tree-like provability in \textsc{NM}$%
_{\rightarrow }$, which follows e.g. from a generalization of Example 8
handling horizontal compression of the HAMILTON PATH problem (see \cite{GH3}%
).

It is well-known that $\mathbf{NP=coNP}$ is a consequence of $\mathbf{%
NP=PSPACE}$ (cf. e.g. \cite{arora}, \cite{Papa}). Nonetheless, for brevity
we also presented a simplified proof of $\mathbf{NP=coNP}$ that does not
refer to \textsc{LM}$_{\rightarrow }$. Instead, we considered the
NP-complete HAMILTON PATH problem and formalized its solvability directly in 
\textsc{NM}$_{\rightarrow }$ (see \cite{GH3} for details).

\section{\protect\small index}

{\small $
\begin{array}{lll}
\text{TM} & \text{turing machine} & \text{\S 1.1} \\ 
\text{SAT, VAL} & \text{satisfacton, validity} & \text{\S 1.1} \\ 
\text{SC} & \text{sequent calculus} & \text{\S 1.1, \S 1.2 } \\ 
\text{ND} & \text{natural deduction} & \text{\S 1.1, \S 1.2} \\ 
\text{normalization},\text{(weak) subformula property} & \text{well-known
optimizations} & \text{\S 1.2} \\ 
\text{(propositional) minimal logic} & \text{standard presentations} & \text{%
\S 2, \S 2.1} \\ 
\text{proof compression} & \text{basic ND operation} & \text{\S 1.1, \S 2.3,
\S 3.1 } \\ 
\text{\textsc{NM}}_{\rightarrow }\text{, \textsc{NM}}_{\rightarrow }^{+} & 
\text{basic ND} & \text{\S 1.2} \\ 
\text{\textsc{NM}}_{\rightarrow }^{\text{\textsc{t}}} & \text{tree-like ND}
& \text{\S 2.1} \\ 
\text{\textsc{NM}}_{\rightarrow }^{\text{\textsc{d}}} & \text{dag-like ND} & 
\text{\S 2.2} \\ 
\begin{array}{l}
\partial \text{, }\partial _{f}\text{, }r\text{, }\rho =\ell \left( r\right) 
\text{ }V,E\text{, }h(-)\text{, }h\left( -\right) \\ 
A\left( -\right) \text{, }A_{f}\left( -\right) \text{, }IV\left( -\right) 
\text{, }IE\left( -\right) \text{, }OV\left( -\right) \text{, }OE\left(
-\right)
\end{array}
& \text{basic ND abbreviations} & \text{\S 2.2} \\ 
\left\lceil x\right\rceil \text{, }\left( \rightarrow I\right) \text{, }%
\left( \rightarrow E\right) \text{, }\left( R\right) _{n}\text{, }\left(
R\right) \text{, }\left( S\right) & \text{basic ND inferences} & \text{\S
2.1, \S 2.2} \\ 
\text{(discharged) assumptions},\text{parents},\text{children} & \text{basic
ND notions} & \text{\S 2.1},\text{\S 2.2} \\ 
\text{maximal (regular) deductive path} & \text{-----''-----} & \text{\S 2.2}
\\ 
\text{regular deduction, proof, }\partial _{f}\vdash \rho & \text{basic 
\textsc{NM}}_{\rightarrow }^{\text{\textsc{d}}}\text{ notions} & \text{\S 2.2%
} \\ 
\text{flow function $f$, $f$-thread} & \text{-----''-----} & \text{\S 2.2}
\\ 
C\text{, }C_{0}\text{, }C_{f} & \text{circuits, Lemma 9} & \text{\S 2.2} \\ 
\text{(quasi-)polynomial}\ \text{deductions} & \text{Definition10} & \text{%
\S 2.2} \\ 
\text{main theorem}\ \text{deductions} & \text{Theorem 11} & \text{\S 2.2}
\\ 
\text{weight, size, }h\left( \partial \right) \text{, }\phi \left( \partial
\right) & \text{deduction parameters} & \text{\S 2.3} \\ 
\text{local coherency } & \text{horizontal compression} & \text{\S 2.3} \\ 
\text{fundamental set of paths (\emph{fsp})} & \text{-----''-----, Lemma 12}
& \text{\S 2.3} \\ 
\mathcal{F}^{\text{\textsc{d}}}\text{, }\mathcal{F}^{\flat }\text{, }%
\partial ^{\text{\textsc{d}}}\text{, }\partial _{f}^{\ast }\text{, }\partial
^{\flat }\text{, }\partial _{f}^{\flat } & \text{-----''-----, Lemma 12} & 
\text{\S 2.3} \\ 
\text{\textsc{LM}}_{\rightarrow } & \text{minimal SC} & \text{\S 3.1 }
\end{array}
$ }


\begin{thebibliography}{99}
\bibitem{arora}  {\small S.\negthinspace\ Arora, B.\negthinspace\ Barak, 
\textbf{Computational Complexity:\negthinspace\ A \negthinspace Modern
\negthinspace Approach}, Cambridge University Press (2009)}

\bibitem{CN}  {\small S. Cook, P. Nguen, \textbf{Logical Foundation of Proof
Theory}, ASL, Cambridge University Press (2010) }

\bibitem{Fel}  {\small Walter Felscher, \textbf{Lectures on Mathematical
Logic. Vol II: Calculi for Derivations and Deductions},Gordon and Breach
Science Publishers (2000) }

\bibitem{Gen1}  {\small G. Gentzen, \emph{Untersuchungen \"{u}ber das
logische Schliessen !, II}, Mathematische Zeitschrift (39): 175--210,
405--431 (1935)}

\bibitem{Gen2}  {\small G. Gentzen, \textbf{The collected Papers of Gerhard
Gentzen}, North-Holland, Amsterdam (1969)}

\bibitem{GH1}  {\small L. Gordeev, E. H. Haeusler, \emph{Proof Compression
and NP Versus PSPACE}, Studia Logica (107) (1): 55--83 (2019) }

\bibitem{GH2}  {\small L. Gordeev, E. H. Haeusler, \emph{Proof Compression
and NP Versus PSPACE II}, Bulletin of the Section of Logic (49) (3):
213--230 (2020), http://dx.doi.org/10.18788/0138-0680.2020.16 }

\bibitem{GH3}  {\small L. Gordeev, E. H. Haeusler, \emph{Proof Compression
and NP Versus PSPACE II: Addendum}, Bulletin of the Section of Logic (51), 9
pp. (2022) http://dx.doi.org/10.18788/0138-0680.2022.01L }

\bibitem{Haeus}  {\small E. H. Haeusler, \emph{Propositional Logics
Complexity and the Sub-Formula Property}, in Proceedings Tenth International
Workshop on Developments in Computational Models, {DCM} 2014, Vienna,
Austria, 13th July 2014. }

\bibitem{Haeus1}  {\small E. H. Haeusler, J.F.C. Barros Jr, R.C.M. Brasil
Filho, \emph{On the horizontal compression of dag-derivations in minimal
purely implicational logic}, arXiv: 2206.02300v4 (2023) }

\bibitem{HB1}  {\small D. Hilbert and P. Bernays, \textbf{Grundlagen der
Mathematik, Bd I}, Springer-Verlag (1934), 2nd edition (1968)}

\bibitem{HB2}  {\small D. Hilbert and P. Bernays, \textbf{Grundlagen der
Mathematik, Bd II}, Springer-Verlag (1939), 2nd edition (1970)}

\bibitem{Hud}  {\small J. Hudelmaier, \emph{An }$O\left( n\log n\right) $%
\emph{-space decision procedure for intuitionistic propositional logic}, J.
Logic Computat. (3): 1--13 (1993) }

\bibitem{Jer}  {\small E. Je\v{r}\'{a}bek, \emph{A simplified lower bound
for implicational logic}, Cambridge University Press, Bulletin of Symbolic
Logic 31 (1): 53--87 (2005) }

\bibitem{Joh}  {\small I. Johansson, \emph{Der Minimalkalk\"{u}l, ein
reduzierter intuitionistischer Formalismus}, Compositio Mathematica (4):
119--136 (1936) }

\bibitem{Neu}  {\small J. von Neumann, \emph{Zur Hilbertschen Beweistheorie}%
, Mathematische Zeitschrift (26): 1--46 (1927) }

\bibitem{Papa}  {\small C. H. Papadimitriou, \textbf{Computational Complexity%
}, Addison-Wesley PC (1995) }

\bibitem{Prawitz}  {\small D. Prawitz, \textbf{Natural deduction: a
proof-theoretical study}, Almqvist \& Wiksell, 1965; Dover Publications,
2006 }

\bibitem{PraMa}  {\small D. Prawitz, P.-E. Malmn\"{a}s, \emph{A survey of
some connections between classical, intuitionistic and minimal logic},
Studies in logic and the Foundations of Math. (50): 215--229 (1968) }

\bibitem{Rob}  {\small Robinson Callou,\newline
https://github.com/RCMBF/Horizontal-Compression/tree/main}

\bibitem{Statman}  R. Statman, \emph{Intuitionistic propositional logic is
polynomial-space complete}, Theor. Comp. Sci. (9): 67--72 (1979)

\bibitem{Sch1}  {\small K. Sch\"{u}tte, \textbf{Beweistheorie},
Springer-Verlag \ (1960)}

\bibitem{Sch2}  {\small K. Sch\"{u}tte, \textbf{Proof Theory},
Springer-Verlag \ (1977)}

\bibitem{Svejdar}  {\small V. \^{S}vejdar, \emph{On the polynomial-space
completeness of intuitionistic propositional logic}, Archive for Math. Logic
(42): 711--716 (2003) }

\bibitem{Tar}  {\small A. Tarski, \textbf{Logic, Semantics, Metamathematics.
Papers from 1923 to 1938}, Clarendon Press, Oxford, U.K. (1956) }

\bibitem{TG}  {\small A. Tarski and S. Givant, \textbf{A Formalization of
Set Theory without Variables}, AMS Colloquium Publications 41 (1987) }

\bibitem{TS}  {\small A. S. Troelstra and H. Schwichtenberg, \textbf{Basic
Proof Theory}, Cambridge University (1996) }
\end{thebibliography}
\end{document}